\documentclass[aps, superscriptaddress]{revtex4-2}

\usepackage{amsmath,amssymb,amsfonts}
\usepackage{algorithmic}
\usepackage{graphicx}
\usepackage{textcomp}
\usepackage{subfig}
\usepackage{mathrsfs}
\usepackage{verbatim} 
\makeatletter
\def\maketag@@@#1{\hbox{\m@th\normalfont\normalsize#1}}
\makeatother
\usepackage{mathtools}
\usepackage{physics}
\usepackage{graphicx}
\usepackage{physics}
\usepackage{tikz}
\usetikzlibrary{quantikz2}
\usetikzlibrary{arrows.meta,
	chains,
	positioning,
	quotes,
	shapes.geometric}
\usepackage[colorlinks=true,citecolor=blue,linkcolor=magenta]{hyperref}
\makeatletter
\renewcommand{\qwbundle}[2][]{%
	\pgfkeys{/quantikz/gates/.cd,style=,Strike Width=0.08cm,Strike Height=0.12cm,#1}%
	\pgfkeysgetvalue{/quantikz/gates/style}{\qz@style}%
	\pgfkeysgetvalue{/quantikz/gates/Strike Width}{\qz@sw}%
	\pgfkeysgetvalue{/quantikz/gates/Strike Height}{\qz@sh}%
	\expanded{%
		\noexpand\arrow[strike arrow={\qz@sw}{\qz@sh}{\unexpanded{#2}},\qz@style,phantom]{l}%
	}%
}
\makeatother

\usepackage{pgfplots}
\usepackage{filecontents}

\usepackage{array}
\usepackage{dcolumn}
\usepackage{bm}
\usepackage{enumitem} 
\usepackage{amsmath}
\usepackage{amssymb}
\usepackage{amsthm}
\usepackage[export]{adjustbox}
\newtheorem{theorem}{Theorem}

\newtheorem{corollary}[theorem]{Corollary}
\newtheorem{definition}[theorem]{Definition}
\newtheorem{proposition}[theorem]{Proposition}

\newcommand\numberthis{\addtocounter{equation}{1}\tag{\theequation}}
\usepackage{listings}
\definecolor{codegreen}{rgb}{0,0.6,0}
\definecolor{codegray}{rgb}{0.5,0.5,0.5}
\definecolor{codepurple}{rgb}{0.58,0,0.82}
\definecolor{backcolour}{rgb}{0.97,0.97,0.97}
\lstdefinestyle{mystyle}{
	backgroundcolor=\color{backcolour},   
	commentstyle=\color{codegreen},
	keywordstyle=\color{blue},
	numberstyle=\tiny\color{codegray},
	stringstyle=\color{codepurple},
	basicstyle=\ttfamily\scriptsize,
	breakatwhitespace=false,         
	breaklines=true,                 
	captionpos=b,                    
	keepspaces=true,                 
	numbers=left,                    
	numbersep=5pt,                  
	showspaces=false,                
	showstringspaces=false,
	showtabs=false,                  
	tabsize=2
}

\newcolumntype{P}[1]{>{\centering\arraybackslash}p{#1}}
\newcolumntype{M}[1]{>{\centering\arraybackslash}m{#1}}

\newcommand{\ou}[1]{{\color{orange} \textsf{OU:} #1}}

\usetikzlibrary{arrows.meta,
	chains,
	positioning,
	quotes,
	shapes.geometric}
\tikzset{FlowChart/.style={
		startstop/.style = {rectangle, rounded corners, draw, fill=red!30,
			minimum width=3cm, minimum height=1cm, align=center,
			on chain, join=by arrow},
		process/.style = {rectangle, draw, fill=orange!30,
			text width=5cm, minimum height=1cm, align=center,
			on chain, join=by arrow, text width = 10cm},
		decision/.style = {diamond, aspect=1.5, draw, fill=green!30,
			minimum width=3cm, minimum height=1cm, align=center,
			on chain, join=by arrow},
		io/.style = {trapezium, trapezium stretches body,   
			trapezium left angle=70, trapezium right angle=110,
			draw, fill=blue!30,
			minimum width=3cm, minimum height=1cm,
			text width =\pgfkeysvalueof{/pgf/minimum width}-2*\pgfkeysvalueof{/pgf/inner xsep},
			align=center,
			on chain, join=by arrow, text width = 10cm},
		arrow/.style = {thick,-Triangle}
	}
}

\usepackage[linesnumbered, ruled, vlined]{algorithm2e}
\usepackage{bm}

\begin{document}
	
	\title{Practical Quantum Circuit Implementation for Simulating Coupled Classical Oscillators}

	\author{Natt Luangsirapornchai}\thanks{Equal contribution}
	\affiliation{Department of Computer Engineering, Faculty of Engineering, Chulalongkorn University, Bangkok 10330, Thailand}
	
	\author{Peeranat Sanglaor}\thanks{Equal contribution}
	\affiliation{Chula Intelligent and Complex Systems Lab, Department of Physics, Faculty of Science, Chulalongkorn University, Bangkok 10330, Thailand}
	
	\author{Apimuk Sornsaeng}
	\affiliation{Chula Intelligent and Complex Systems Lab, Department of Physics, Faculty of Science, Chulalongkorn University, Bangkok 10330, Thailand}
	\affiliation{Science, Mathematics and Technology Cluster, Singapore University
		of Technology and Design, 8 Somapah Road, 487372 Singapore}
	
	\author{St\'{e}phane Bressan}
	\affiliation{School of Computing, National University of Singapore, 13 Computing Drive, 117417, Singapore, Singapore}
	
	\author{Thiparat Chotibut}
	\email[Correspondence to: ]{thiparatc@gmail.com}
	\affiliation{Chula Intelligent and Complex Systems Lab, Department of Physics, Faculty of Science, Chulalongkorn University, Bangkok 10330, Thailand}
	
	\author{Kamonluk Suksen}
	\email[Correspondence to: ]{kamonluk@cp.eng.chula.ac.th}
	\affiliation{Department of Computer Engineering, Faculty of Engineering, Chulalongkorn University, Bangkok 10330, Thailand}
	
	\author{Prabhas Chongstitvatana}\email[Correspondence to: ]{ prabhas.c@chula.ac.th}
	\affiliation{Department of Computer Engineering, Faculty of Engineering, Chulalongkorn University, Bangkok 10330, Thailand}

	\date{\today}
	
	\begin{abstract}
		Simulating large-scale coupled-oscillator systems presents substantial computational challenges for classical algorithms, particularly when pursuing first-principles analyses in the thermodynamic limit. Motivated by the quantum algorithm framework proposed by Babbush \textit{et al.} \cite{Babbush_2023}, we present and implement a detailed quantum circuit construction for simulating one-dimensional spring-mass systems. Our approach incorporates key quantum subroutines, including block encoding, quantum singular value transformation (QSVT), and amplitude amplification, to realize the unitary time-evolution operator associated with simulating classical oscillators dynamics. In the uniform spring-mass setting, our circuit construction requires a gate complexity of \(\mathcal{O}\bigl(\log_2^2 N\,\log_2(1/\varepsilon)\bigr)\), where \(N\) is the number of oscillators and \(\varepsilon\) is the target accuracy of the approximation. For more general, heterogeneous spring-mass systems, the total gate complexity is \(\mathcal{O}\bigl(N\log_2 N\,\log_2(1/\varepsilon)\bigr)\). Both settings require \(\mathcal{O}(\log_2 N)\) qubits. Numerical simulations agree with classical solvers across all tested configurations, indicating that this circuit-based Hamiltonian simulation approach can substantially reduce computational costs and potentially enable larger-scale many-body studies on future quantum hardware. 
	\end{abstract}
	
	
	\maketitle
	
	\section{Introduction} \label{sec:background}
	
	\begin{figure*}
		\centering
		\includegraphics[width=0.95\linewidth]{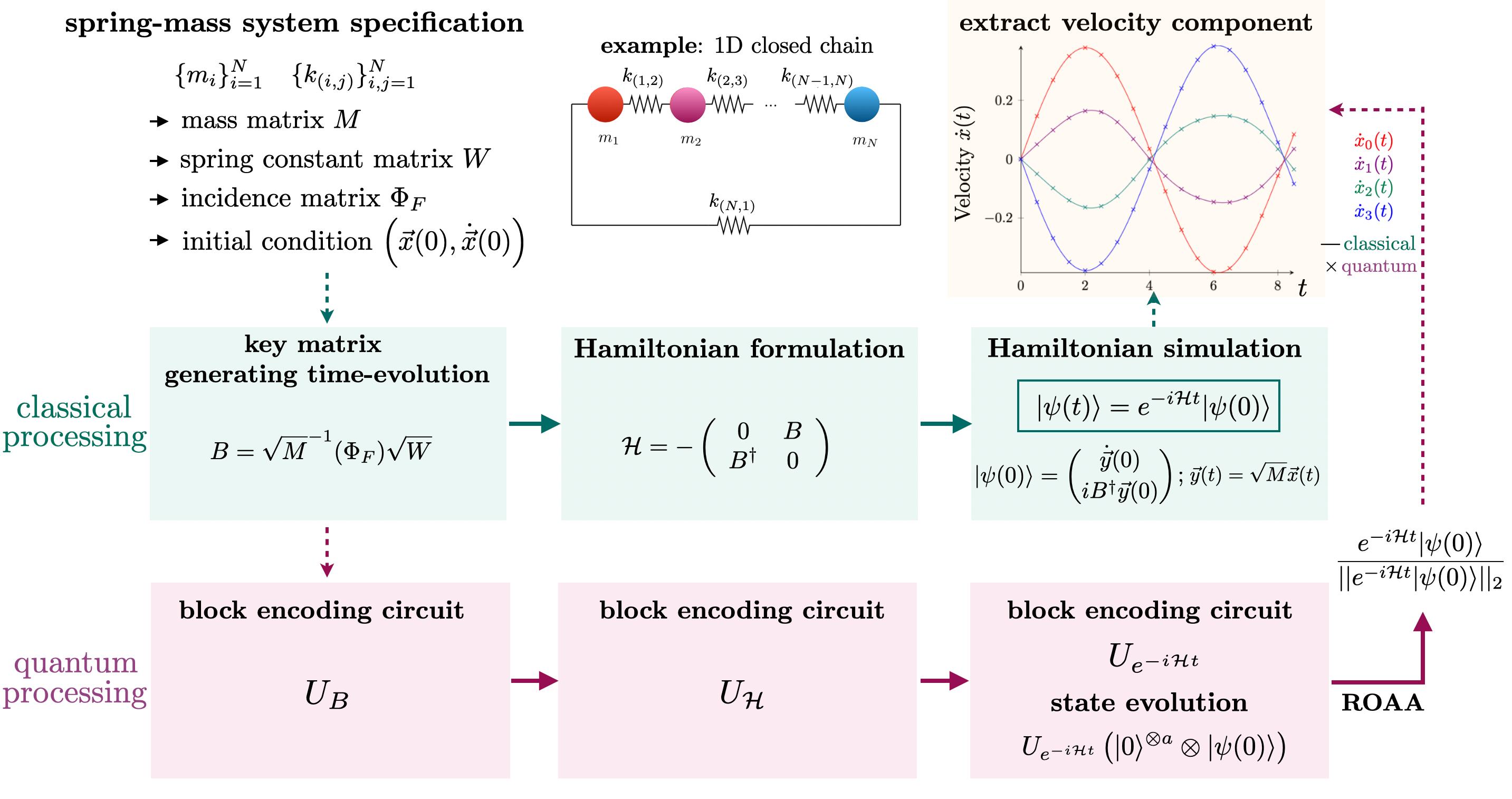}
		\caption{{\bf A high-level workflow illustrating our circuit-based quantum algorithm for simulating a one-dimensional chain of coupled oscillators.} 
			(1) Define the classical system, specifying the initial condition $\left(\vec{x}(0),\dot{\vec{x}}(0)\right)$, masses \(\{m_i\}_{i=1}^N\) and spring constants \(\{k_{(i,j)}\}_{i,j=1}^{N}\), whose connectivity structure can be encoded by the incidence matrix $\Phi_F$.
			(2) Following  \cite{Babbush_2023}, rewrite the classical Newtonian equations of motion in a matrix form and then reformulate them as a Hamiltonian simulation problem with associated Hamiltonian \(\mathcal{H}\), using mass‐weighted state variable \(\vec{y}(t)=\sqrt{M}\,\vec{x}(t)\) (Section \ref{sec:CCO}). 
			(3) Build a block encoding circuit of \(\mathcal{H}\) (via block encoding circuit of \(\sqrt{M}^{-1}\), \(\Phi_F\), \(\sqrt{W}\)) (Section \ref{subsec:BE_B} and \ref{subsec:BE_H}).
			(4) Apply QSVT and linear‐combination‐of‐unitaries (LCU) techniques to implement a block encoding circuit of \(e^{-i\mathcal{H}t}\) (Section \ref{subsec:BE_eiht}). 
			(5) Robust oblivious amplitude amplification (ROAA) boosts the success probability of measuring the correct time-evolved state. (Section \ref{subsec:fullHS}).
			(6) Finally, we extract the time‐evolved state vector, read off velocities directly, and (optionally) reconstruct displacements via a finite‐difference procedure.  
			Classical and quantum results are then compared in post‐processing (Section \ref{sec:results_exp}).}
		\label{fig:diagram}
	\end{figure*}
	
	Simulating the dynamics of large-scale coupled oscillators, ubiquitous in science and engineering, quickly becomes intractable on classical hardware as system size grows large. Recent work has shown that even a seemingly straightforward task of spring-mass simulation can be BQP-complete \cite{Babbush_2023}, suggesting that classical methods may be fundamentally limited for certain instances. This motivates the search for quantum algorithms capable of using quantum hardware to handle such oscillator networks in polynomial or near-polynomial time.
	
	Indeed, a spectrum of quantum algorithms has been proposed to tackle a variety of tasks, from optimization \cite{Abbas_2024}, to quantum chemistry simulations \cite{Cao_2019, Bauer_2020}, and quantum machine learning  \cite{Biamonte_2017}. However, many such methods present high-level formulations (e.g., trotterization, linear-combination-of-unitaries, qDRIFT) and often do not exhaustively detail their gate-level instantiations. In particular, the literature on explicitly implementing the time-evolution of large-scale, classical spring–mass systems typically offers conceptual resource estimates or broad procedures, rather than fully worked-out circuit decompositions.
	
	In this work, we bridge that gap by building on the  framework by Babbush \emph{et al.} \cite{Babbush_2023} to provide a \emph{complete circuit-level implementation} of Hamiltonian simulation for one-dimensional spring-mass systems. Our construction relies on the notion of \emph{block encoding} \cite{Gily_n_2019}, an approach that embeds non-unitary operators into suitably enlarged unitaries. It further exploits the versatility of \emph{quantum singular value transformation} (QSVT) to implement the time-evolution operator $e^{-i \mathcal{H}t}$ with a cost that scales polynomially in $\alpha t$ and $\log\bigl(1/\varepsilon\bigr)$, where $\alpha$ is the sub-normalization constant of the Hamiltonian $\mathcal{H}$ and $\varepsilon$ is the desired approximation error \cite{Gily_n_2019}. While conventional algorithms such as 
	trotterization~\cite{Lloyd1996UniversalQS}, qDRIFT~\cite{Campbell2019}, or 
	quantum phase estimation have proven effective for various Hamiltonian simulations, 
	QSVT offers more refined polynomial approximations and unifies multiple subroutines (\emph{e.g.}, Grover search, quantum linear systems, and Gibbs 
	sampling) into a cohesive framework~\cite{Martyn2021, Low_2017}. Here, we integrate QSVT with \emph{robust oblivious amplitude amplification} (ROAA) \cite{Berry_Childs_Cleve_Kothari_Somma_2014}, ensuring that, with high probability, the ancilla measurement projects onto the correct time-evolved oscillator state.
	
	Practically, our algorithm proceeds in three main stages: (1) recasting the classical oscillator chain’s equations of motion into a block-encoded Hamiltonian simulation formalism, (2) applying QSVT-based polynomial approximations to realize $e^{-i\mathcal{H}t}$, and (3) extracting the velocities from the resulting quantum state and post-processing the displacements. We deploy \texttt{QSPPACK} \cite{Dong_2021, Wang_2022, dong2022infinite} to compute the necessary QSVT phase angles, and then benchmark our quantum outputs against classical solvers (e.g., Runge–Kutta) over a range of parameter regimes including uniform versus heterogeneous masses and springs. Across all tested examples, our circuit-based approach accurately reproduces the oscillator trajectories.
	
	The paper is organized as follows. In Section~\ref{sec:CCO}, we show how the dynamics of classical coupled oscillators can be recast as a quantum Hamiltonian simulation problem. Section~\ref{sec:BE} reviews the core subroutines: block encoding, QSVT , and ROAA. In Section~\ref{sec:result_theory}, we detail how these ingredients combine to yield a coherent quantum circuit for constructing the time-evolution operator and provide gate-complexity estimates. Numerical simulations of various oscillator setups, demonstrating the accuracy of our approach, are presented in Section~\ref{sec:results_exp}. Finally, Section~\ref{sec:conc} concludes our work, discusses the limitations, and outlines future research directions.

	
	\section{Classical Coupled-Oscillator Dynamics in a Hamiltonian Simulation Framework}
	\label{sec:CCO}
	
	We focus on a one-dimensional chain of \(N\) coupled oscillators, each labeled by an index \(j \in \{1,\dots,N\}\). For any pair of adjacent oscillators \(j\) and \(j\pm 1\), a harmonic spring with constant \(k_{(j,j\pm 1)}\) mediates their interaction. We consider two primary boundary conditions:
	\begin{enumerate}
		\item \emph{Periodic boundary condition (closed-chained)}: A spring of constant \(k_{(1,N)}\)  connects the first and last oscillators.
		\item \emph{Open boundary condition (open-chained)}: No spring links oscillator 1 and \(N\), effectively setting \(k_{(1,N)}=0\).
	\end{enumerate}
	
	\subsection{Equations of Motion and Matrix Formulation}
	
	Let \(\vec{x}(t) = \bigl(x_1(t),\,x_2(t),\dots,x_N(t)\bigr)^\mathsf{T}\) denote the instantaneous displacements of the oscillators at time \(t\). Each oscillator \(j\) has mass \(m_j\) and is displaced by \(x_j(t)\) from its equilibrium. The classical equation of motion (EOM) for \(x_j(t)\) is
	\begin{equation}
		m_j\,\ddot{x}_j(t)
		\;=\;
		-\,k_{(j,j-1)}\bigl[x_j(t) - x_{j-1}(t)\bigr]
		\;-\;
		k_{(j,j+1)}\bigl[x_j(t) - x_{j+1}(t)\bigr],
		\label{eq:eom}
	\end{equation}
	where each \(k_{(i,j)}>0\) unless no spring is present, in which case \(k_{(i,j)}=0\). Collecting all displacements into the vector \(\vec{x}(t)\), we define the diagonal \emph{mass matrix} 
	\begin{equation}\label{eq:mass_matrix}
		M
		\;\coloneqq\;
		\mathrm{diag}\bigl(m_1,\,m_2,\dots,m_N\bigr).
	\end{equation}
	The EOMs for all oscillators can be compactly written in matrix form as
	\begin{equation}
		M\,\ddot{\vec{x}}(t)
		\;=\;
		-\,F\,\vec{x}(t),
		\label{eq:EOM}
	\end{equation}
	where \(F\) is the \emph{spring matrix}. The negative of each spring constant appears off-diagonal, and the diagonal entries capture the local sum of connected springs. Specifically,
	\begin{eqnarray}
		\label{eq:cyclic_CO_F}
		F \;\coloneqq\;
		\begin{pmatrix}
			k_{(1,N)} + k_{(1,2)} & -\,k_{(1,2)}         & 0                & \cdots & -\,k_{(1,N)} \\[3pt]
			-\,k_{(1,2)}         & k_{(1,2)} + k_{(2,3)}& -\,k_{(2,3)}     & \cdots & 0            \\[3pt]
			0                    & -\,k_{(2,3)}        & \ddots           & \ddots & \vdots       \\[3pt]
			\vdots               & \vdots              & \ddots           & k_{(N-2,N-1)} + k_{(N-1,N)} & -\,k_{(N-1,N)} \\[3pt]
			-\,k_{(1,N)}         & 0                   & \cdots           & -\,k_{(N-1,N)} & k_{(N-1,N)} + k_{(1,N)}
		\end{pmatrix}.
	\end{eqnarray}
	For the open boundary condition, one sets \(k_{(1,N)}=0\), eliminating the top-right and bottom-left entries.
	
	A common way to numerically solve ~\eqref{eq:EOM} on a classical computer is via standard time-stepping (e.g., fourth-order Runge–Kutta). While RK4 achieves \(\mathcal{O}(\Delta t^5)\) local error per step, it demands \(\mathcal{O}\!\bigl(N\,(t_f-t_i)/\Delta t\bigr)\) time complexity and can become burdensome for large \(N\) (thermodynamics limit) or extended simulation intervals.
	
	\subsection{Mapping to a Hamiltonian Simulation Formalism}
	
	Following Babbush \emph{et al.}~\cite{Babbush_2023}, we now reformulate these classical EOMs into a framework reminiscent of the Schrödinger equation. Define
	\begin{equation}
		\vec{y}(t)
		\;\coloneqq\;
		\sqrt{M}\,\vec{x}(t),
		\label{eq:CO_y}
	\end{equation}
	so that each component of \(\vec{y}\) is a mass-weighted displacement. Differentiating, one finds
	\[
	\ddot{\vec{y}}(t)
	\;=\;
	\mathcal{A}\,\vec{y}(t),
	\quad
	\]
	where 
	\begin{equation} \label{eq:CO_A}
		\mathcal{A}
		\;\coloneqq\;
		\sqrt{M}^{-1}\,F\,\sqrt{M}^{-1}.
	\end{equation}
	Here, \(\mathcal{A}\) is a rescaled version of the spring matrix \(F\). To recast this system as a Hamiltonian simulation, define the combined state vector
	\begin{equation}
		\label{eq:CO_state}
		\ket{\psi(t)}
		\;\;\propto\;
		\begin{pmatrix}
			\dot{\vec{y}}(t)\\[3pt]
			i\,B^\dagger\,\vec{y}(t)
		\end{pmatrix},
	\end{equation}
	where \(B\) satisfies \(B\,B^\dagger = \mathcal{A}\). By construction, \(\ket{\psi(t)}\) obeys an equation structurally analogous to a solution of the Schrödinger equation
	\begin{equation}
		\label{eq:CO_state_sol}
		\begin{pmatrix}
			\dot{\vec{y}}(t)\\[4pt]
			i\,B^\dagger\,\vec{y}(t)
		\end{pmatrix}
		\;=\;
		e^{-\,i\,\mathcal{H}\,t/\hbar}
		\begin{pmatrix}
			\dot{\vec{y}}(0)\\[4pt]
			i\,B^\dagger\,\vec{y}(0)
		\end{pmatrix},
	\end{equation}
	where the Hamiltonian \(\mathcal{H}\) has a block-off-diagonal structure
	\begin{equation}
		\mathcal{H}
		\;=\;
		-\,
		\begin{pmatrix}
			0 & B\\[2pt]
			B^\dagger & 0
		\end{pmatrix}.
		\label{eq:CO_H}
	\end{equation}
	Namely, \(\ket{\psi(t)}\) evolves according to a linear first-order differential equation in the same form as the Schrödinger equation:
	\(\displaystyle \partial_t \ket{\psi(t)} = -\,i\,\mathcal{H}\,\ket{\psi(t)}/\hbar.\)
	Hence, its formal solution is 
	\(\displaystyle \ket{\psi(t)} = e^{-\,i\,\mathcal{H}\,t/\hbar}\ket{\psi(0)},\)
	the standard time-evolution in the Schrödinger picture.
	In the following, we adopt natural units (\(\hbar=1\)) and omit the explicit \(\hbar\) in the exponent.
	
	From a quantum simulation standpoint, ~\eqref{eq:CO_state_sol} indicates that classical spring–mass dynamics can be recast as a time-evolution operator \(e^{-\,i\,\mathcal{H}\,t}\). Realizing this time-evolution operator on a quantum device forms a basis of our proposed algorithm. In the next sections (particularly Section~\ref{sec:BE}), we detail how techniques like \emph{block encoding} and \emph{quantum singular value transformation} (QSVT) can implement \(e^{-\,i\,\mathcal{H}\,t}\) in a quantum circuit.

	\subsection{Graph-Laplacian Perspective and the Construction of Matrix \boldmath\(B\)}
	
	From ~\eqref{eq:cyclic_CO_F}, one observes that each row of the spring matrix \(F\) sums to the total spring constant connecting a given oscillator to its neighbors, while off-diagonal entries \(-\,k_{(i,j)}\) capture the pairwise couplings. Consequently, \(F\) naturally takes the form of a weighted \emph{graph Laplacian}. Indeed, if we view the oscillator indices as vertices in a graph, then \(k_{(i,j)}\) corresponds to an edge weight between vertices \(i\) and \(j\). Concretely, one may factorize \(F\) via
	\[
	F
	\;=\;
	\Phi_F\,W\,\Phi_F^\dagger,
	\]
	where \(\Phi_F\) is the \emph{incidence matrix} of this weighted graph, and \(W\) is the diagonal matrix collecting all relevant spring constants given by
	\begin{equation}\label{eq: W_matrix}
		W
		\;\coloneqq\;
		\mathrm{diag}\bigl(k_{(1,2)},\,k_{(2,3)},\,\ldots,\,k_{(N-1,N)},\,k_{(1,N)}\bigr).
	\end{equation}
	Recalling ~\eqref{eq:CO_A}, we now explicitly define the \emph{key matrix} $B$ satisfying
	$B\,B^\dagger = \mathcal{A}$ that essentially encapsulates the \emph{mass-spring network
		coupling structure} as 
	\begin{equation}
		\label{eq:CO_weightB}
		B
		\;\coloneqq\;
		\sqrt{M}^{-1}\,\Phi_F\,\sqrt{W}.
	\end{equation}
	In a system with all masses and spring constants equal to 1, \(\sqrt{M}=I\) and \(\sqrt{W}=I\), making \(B=\Phi_F\). 
	
	It is important to emphasize that \(B\) constitutes the block-off-diagonal structure of the Hamiltonian \(\mathcal{H}\) in ~\eqref{eq:CO_H}. Thus, any quantum simulation of \(e^{-\,i\mathcal{H}t}\) must efficiently encode \(B\) into a larger unitary, a task we accomplish using \emph{block encoding}. In the sections that follow, we outline how this approach accommodates both uniform and nonuniform parameter choices, as well as open versus periodic boundary conditions.

	\section{Block Encoding}
	\label{sec:BE}
	
	In the preceding section, we encountered matrices that are non-unitary and thus cannot be directly realized as quantum gates (which must be unitary). A useful technique for embedding arbitrary linear operators into larger unitaries is known as \emph{block encoding} \cite{Gily_n_2019}. Below, we formalize the concept of a block encoding and show how it supports subsequent circuit constructions for time-evolution operators $e^{-\,i\,\mathcal{H}\,t}$.
	
	\begin{definition}[Block Encoding]\label{def:BE}
		Let \(A\) be an operator acting on \(s\) qubits (i.e., an \(s\)-qubit matrix). An \((\alpha,a,\varepsilon)\)\emph{-block-encoding} of \(A\) is a \((s+a)\)-qubit unitary \(U_A\) that has the block structure
		\begin{equation} \label{eq:BE_mat}
			U_A = \mqty(A/\alpha 	& 	\ast \\ 
			\ast	&	\ast)
		\end{equation}
		such that
		\begin{equation} \label{defEq_BE}
			\norm{A - \alpha(\bra{0^a} \otimes I_{s}) U_A (\ket{0^a} \otimes I_{s})}_2 \leq \varepsilon \ ,
		\end{equation}
		where \(I_s\) is the identity on the \(s\) signal qubits, \(\ket{0^a}\equiv \ket{0}^{\otimes a}\) are \(a\) ancillary qubits, \(\alpha\ge 0\) is called the \emph{sub-normalization constant}, and \(\varepsilon\ge 0\) is the encoding error. The symbol \(\ast\) indicates unspecified blocks that make \(U_A\) unitary. 
	\end{definition}
	
	In essence, \(\ket{0^a}\) acts as a ``flag'' indicating when the operator \(A\) is successfully applied: measuring all ancillary qubits in \(\ket{0^a}\) collapses the circuit to the action of \(A\). A schematic of such a \((\alpha,a,\varepsilon)\)-block-encoding is shown in Fig.~\ref{fig:BE_circuit_schematic}. Note that the error of encoding and the sub-normalization constant must follow the condition \cite{Gily_n_2019}:
	\[
	\|A\|_2 \;\le\; \alpha \;+\;\varepsilon,
	\]
	in keeping with the block structure in ~\eqref{eq:BE_mat}. Once a block-encoding is available, subsequent composition techniques such as tensor products, multiplications, or linear combinations of these block-encoded operators (often referred to as the linear combination of unitaries or LCU) enable rich algebraic manipulations purely at the circuit level (see Appendix).
	
	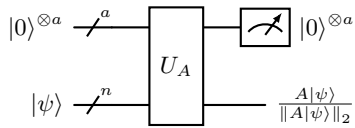
\begin{figure}[ht!]
		\centering
		\begin{quantikz} 
			\lstick{$\ket{0}^{\otimes a}$}	&\qwbundle{a}	&\gate[2]{U_A}		&\meter{} \rstick{$\ket{0}^{\otimes a}$} \\
			\lstick{$\ket{\psi}$}			&\qwbundle{n}	&					&\rstick{$\frac{A\ket{\psi}}{\norm{A\ket{\psi}}_2}$}
		\end{quantikz}
		\caption{{\bf A $(\alpha,a,\varepsilon)$-block-encoding of an $s$-qubit operator $A$.} The top $a$ qubits are ancillary (\emph{ancilla register}), initially in state \(\ket{0^a}\). The bottom $s$ qubits (\emph{signal register}) hold the state on which \(A\) acts. Measuring the ancillas in the state \(\ket{0^a}\) projects onto \(A/\alpha\).}
		\label{fig:BE_circuit_schematic}
	\end{figure}
	
	\subsection{Quantum Singular Value Transformation (QSVT)}
	\label{subsec:QSVT}
	
	Given a block encoding of a Hamiltonian \(\mathcal{H}\), one can construct more complicated operators such as \(e^{-\,i\mathcal{H}\tau}\) by applying \emph{quantum singular value transformation} (QSVT) \cite{Gily_n_2019}. QSVT generalizes techniques from \emph{quantum signal processing} (QSP) \cite{Low_2017} to allow polynomial transformations of singular values (or equivalently, eigenvalues for normal matrices) of the block-encoded operator \cite{Gily_n_2019,Martyn2021,Martyn2023}. 
	
	In particular, if \((\alpha,a,0)\)-block-encodings of \(\mathcal{H}\) are available, one can approximate \(\cos(\mathcal{H}\tau)\) and \(-\,i\,\sin(\mathcal{H}\tau)\) via polynomial expansions, then combine them (via an LCU trick) to form \(\exp(-i\mathcal{H}\tau)\). We often refer to QSVT over real-symmetric or Hermitian matrices as \emph{quantum eigenvalue transformation} (QET) \cite{Martyn2023}, but the term QSVT is used broadly.
	
	The foundation of QSVT lies in QSP, which transforms the real or imaginary part of a unitary matrix into a target polynomial \(P(\cdot)\) via a set of \emph{signal-processor} phase angles \(\vec{\phi} = \{\phi_0, \phi_1, \dots, \phi_d\}\). For instance, in a \(2\times 2\) setting,
	\begin{equation*} 
		U 
		\;=\; 
		\begin{pmatrix}
			a & * \\
			* & *
		\end{pmatrix}
		\xrightarrow[\text{QSP}]{\phantom{X}}
		U^{\vec{\phi}}
		\;=\;
		\begin{pmatrix}
			P(a) & * \\
			* & *
		\end{pmatrix},
	\end{equation*}
	where the angles \(\{\phi_j\}\) dictate how projector-controlled phase gates (or equivalently, \(Z\)-rotations) “process” the matrix entries to yield the desired polynomial transformation. In practice, these \(\phi_j\) can be computed \emph{offline}, for example, via \texttt{QSPPACK} in MATLAB \cite{Dong_2021,Wang_2022,dong2022infinite}. We then embed this QSP subroutine into a broader QSVT circuit by interleaving our signal-processor phase gates with the block-encodings of the target operator (e.g., a Hamiltonian \(\mathcal{H}\)), as described formally in Theorem~\ref{thm:QSVT}. In our work, we adopt and slightly modify the QSVT constructions from Refs.~\cite{lin2022lecture, ralli2023calculating}, to ensure that once the polynomial coefficients are fixed, we can assemble the final time-evolution circuit $e^{-i\mathcal{H}t}$ through a sequence of carefully designed phase gates and block-encoded unitaries.
	
	\begin{definition}[Projector-controlled phase gate]\label{def:QSVT_projConPhase}
		Let $\Pi$ be a projector onto a desired subspace, and let $\phi$ be a real angle. We define the projector-controlled phase-$\phi$ gate as
		\begin{equation}
			\label{eq:HS_projConPhase}
			\Pi(\phi) \;\coloneqq\; e^{\,i\,\phi \,\bigl(2\Pi \,-\,I\bigr)}.
		\end{equation}
		Equivalently, $\Pi(\phi)$ applies a phase $e^{\,i\phi}$ when the ancilla register lies in subspace $\Pi$, and $e^{-\,i\phi}$ otherwise. In the context of block-encodings, one often has $\Pi = \dyad{0^a}$ or a multi-qubit projector indicating $\ket{0^a}$.
	\end{definition}
	
	One may interpret $\Pi(\phi)$ as an operator that rotates by $+\,\phi$ on the projector subspace and by $-\,\phi$ on its orthogonal complement. Practically, one realizes $\Pi(\phi)$ with a $Z$-rotation gate $e^{-\,i\,\phi\,Z}$ (i.e., $R_z(2\phi)$) sandwiched between two $\mathrm{C}_{\Pi}\mathrm{NOT}$ gates, where these modified CNOTs activate only when the ancilla qubits match the projector $\Pi$. Figure~\ref{fig:QSVT_proj_gate} illustrates this construction.
	
	\begin{figure}[ht!]
		\centering
		\begin{quantikz}
			\lstick{$\ket{0}$}		&&\targ{}\wire[d][1]{q} &\gate{e^{-i\phi Z}}	&\targ{}\wire[d][1]{q} & \\
			\lstick{$\ket{\psi_0}$}	&\qwbundle{}&\gate{\Pi}				&					&\gate{\Pi} &
		\end{quantikz}
		\caption{{\bf Circuit for the projector-controlled phase gate $\Pi(\phi)$.}
			An $e^{-\,i\phi Z}$ rotation is flanked by two projector-controlled NOT gates $\mathrm{C}_{\Pi}\mathrm{NOT}$. 
			Multi-controlled gates arise naturally if $\Pi$ projects onto more than one qubit.}
		\label{fig:QSVT_proj_gate}
	\end{figure}
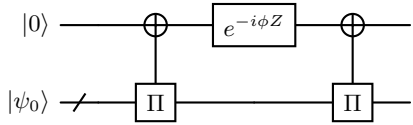
	
	\begin{definition}[C$_{\Pi}$NOT gate]\label{def:QSVT_CpiNOT}
		For the same projector $\Pi$, define
		\begin{equation}\label{eq:QSVT_CpiNOT}
			\mathrm{C}_{\Pi}\mathrm{NOT}
			\;\coloneqq\;
			X\,\otimes\,\Pi
			\;+\;
			I\,\otimes\bigl(I-\Pi\bigr).
		\end{equation}
		This operator flips (NOT) the first qubit precisely when the ancilla qubits lie in subspace $\Pi$. In a simple case where $\Pi = \dyad{00}$, $\mathrm{C}_{\Pi}\mathrm{NOT}$ reduces to a $\mathrm{CCNOT}$ gate. In general, one can realize $\mathrm{C}_{\Pi}\mathrm{NOT}$ as a multi-controlled NOT acting conditionally on the ancilla projector.
	\end{definition}
	
	\noindent
	Together, $\Pi(\phi)$ and $\mathrm{C}_{\Pi}\mathrm{NOT}$ provide a flexible mechanism for embedding phase shifts that depend on an ancilla subspace. In block-encoding applications, these gates enable polynomial transformations on the singular values or eigenvalues of the encoded operator, forming key elements of quantum signal processing and QSVT protocols, which we now explain.

	\begin{theorem}[Quantum Singular Value Transformation]\label{thm:QSVT}
		Let $A$ be an $s$-qubit operator with a $(\alpha,a,0)$--block-encoding $U_A$, meaning
		\[
		\left(\bra{0^a}\otimes I_s\right) \;U_A\;\left(\ket{0^a}\otimes I_s \right) \;\approx\;A/\alpha.
		\]
		Suppose $P(\cdot)$ is a target polynomial, defined on a suitable domain (e.g., $\|A\|_2\le 1$ so that $P(\|A\|_2)\le 1$). Then, there exists a polynomially sized set of phase angles $\{\phi_0,\dots,\phi_d\}$ such that one obtains a block-encoding of $P(A)$ denoted by $U_{A, \mathrm{QSVT}}^{\vec{\phi}}$ through the following construction:
		
		\begin{eqnarray}
			\label{eq:QSVTproc_comp}
			U_{A, \mathrm{QSVT}}^{\vec{\phi}}
			\;=\;
			\Bigl(H\;\otimes\;I_{a+s}\Bigr)
			\;U_A^{\vec{\varphi}}
			\;\Bigl(H\;\otimes\;I_{a+s}\Bigr),
		\end{eqnarray}
		where the core QSVT sequence $U_A^{\vec{\varphi}}$ alternates $U_A$ and $U_A^\dagger$ with projector-controlled phase gates $\Pi(\varphi_j)$ and possibly additional single-qubit rotations. Here, $H$ denote the Hadamard gate, not with the system Hamiltonian $\mathcal{H}$. Concretely,
		
		\begin{eqnarray}
			\label{eq:QET_core}
			U_A^{\vec{\varphi}} =
			\begin{cases}
				R_z^{f_e(d)}(2\pi)\,\Pi(\varphi_0)
				\displaystyle\prod_{j=1}^{d/2}
				\bigl[
				U_A^\dagger\,\Pi(\varphi_{2j-1})\,U_A\,\Pi(\varphi_{2j})
				\bigr]
				& \text{for even $d$,}\\[8pt]
				R_z^{f_o(d)}(2\pi)\,\Pi\bigl(-\varphi_0\bigr)\,U_A
				\displaystyle\prod_{j=1}^{(d-1)/2}
				\bigl[
				\Pi\bigl(-\varphi_{2j-1}\bigr)\,U_A^\dagger\,\Pi\bigl(-\varphi_{2j}\bigr)\,U_A
				\bigr]
				& \text{for odd $d$.}
			\end{cases}
		\end{eqnarray}
		
		\noindent
		Here, $d$ is the degree (or approximate degree) of $P(\cdot)$, and the functions $f_e(d)$, $f_o(d)$ specify whether an extra $R_z$ gate contributes a $\pm 1$ factor depending on $d \bmod 4$. Specifically,
		\begin{eqnarray} \label{eq:QET_f_e}
			f_e(d) =
			\begin{cases}
				0 &\text{if $d\bmod 4 = 0$,} \\
				1 &\text{otherwise,}
			\end{cases}
		\end{eqnarray}
		\begin{eqnarray} \label{eq:QET_f_o}
			f_o(d) =
			\begin{cases}
				0 &\text{if $d\bmod 4 = 1$,} \\
				1 &\text{otherwise,}
			\end{cases}
		\end{eqnarray}
		The phase angles $\varphi_j$ are offset from the original $\phi_j$ returned by, e.g., \texttt{QSPPACK}~\cite{Dong_2021,Wang_2022,dong2022infinite}, according to
		\begin{eqnarray}
			\label{eq:varphiTF}
			\varphi_j \;=\;
			\begin{cases}
				\phi_0 + \tfrac{\pi}{4}, & j=0,\\[4pt]
				\phi_j + \tfrac{\pi}{2}, & j=1,\dots,d-1,\\[4pt]
				\phi_d + \tfrac{\pi}{4}, & j=d.
			\end{cases}
		\end{eqnarray}
	\end{theorem}
	
	\begin{proof}[Sketch of Proof]
		A complete derivation can be found in Ref.~\cite{Dong_2021} (see also \cite{Martyn2021}), but we outline the key idea. In quantum signal processing (QSP), one constructs polynomial transformations of a unitary’s singular (or eigen) values by interleaving projector-controlled phase gates with the block-encoded operator. Our variant modifies the original QSP/QSVT sequence slightly to accommodate sine-function approximations and ensure correct phase alignment. This modification shifts $\phi_0$ and $\phi_d$ by $\tfrac{\pi}{4}$ and each intermediate $\phi_j$ by $\tfrac{\pi}{2}$, so that the resulting sequence indeed implements $P(A)$ up to the desired precision.
	\end{proof}
	
	\paragraph*{Remarks on Constraints.}
	For $U_{A, \mathrm{QSVT}}^{\vec{\phi}}$ to accurately represent $P(A)$:
	\begin{enumerate}[label=(\roman*)]
		\item \emph{Hermiticity.} Typically, $A$ should be Hermitian (or real-symmetric) so that the singular values match its eigenvalues; see \cite{Martyn2021} for details.
		\item \emph{Bounded Domain.} One often restricts $\|A\|_2\le 1$, ensuring that $P(\cdot)$ remains well-defined and $\|P(\cdot)\|_2\le 1$.  
		\item \emph{Polynomial Parity.} For real polynomials, $P(\cdot)$ generally must be even or odd to fit the QSVT subroutine neatly.  
		\item \emph{Amplitude Amplification.} Measuring the ancilla qubits in $\ket{0^a}$ yields the transformed state (or matrix) $P(A)$. Since this measurement may not always succeed with high probability, amplitude amplification or robust oblivious amplitude amplification (ROAA) \cite{Berry_Childs_Cleve_Kothari_Somma_2014} is often employed to boost success probability.
	\end{enumerate}
	
	As a final implementation note, the circuit in Fig.~\ref{fig:QET_mod} (a simplified depiction of QSVT for a Hermitian operator) shows how the block-encoded $A$ is repeatedly interspersed with phases $\Pi(\varphi_j)$. The precise angle set $\{\phi_j\}$ can be computed offline based on the target polynomial $P(\cdot)$ (e.g., using Chebyshev expansions or numerical solvers such as \texttt{QSPPACK}), and the final gate sequence is then assembled into a coherent QSVT transformation.
	
	\begin{figure}[ht!]
		\centering
		\begin{adjustbox}{width=0.9\textwidth}
			\begin{quantikz}
				\lstick{$\ket{0}$}	&\gate{H} && \gate[2]{\Pi(\varphi_d)} & & \gate[2]{\Pi(\varphi_{d-1})} & & \midstick[3, brackets=none]{$\cdots$} & &\gate[2]{\Pi(\varphi_0)} &&\gate{H}&\\
				\lstick{$\ket{0}^{\otimes a}$}&&\qwbundle{a} & &\gate[2]{U_A} & &\gate[2]{U_A^\dagger} & &\gate[2]{U_A} & & &&\\
				\lstick{$\ket{\psi}$}&	&\qwbundle{s} & & & & & & & & &&\\	
			\end{quantikz}
		\end{adjustbox}
		\caption{{\bf A schematic quantum eigenvalue transformation circuit of $s$-qubit operator $A$ with its $(\alpha_A, a, 0)$--block-encoding $U_A$.} By interleaving \(U_A\) and \(U_A^\dagger\) with 
			projector-controlled phase gates, one effectively applies a polynomial \(P(\lambda)\) to 
			the eigenvalues \(\lambda\) of \(A\). Measuring the ancillas in \(\ket{0^a}\) then 
			projects onto the subspace corresponding to \(P(A)\). In standard QSVT, non-Hermitian 
			operators require singular-value transformations, but for Hermitian \(A\), we can speak 
			directly of eigenvalue transformations (QET). Note that $H$ here represents the Hadamard gate.} 
		\label{fig:QET_mod}
	\end{figure}
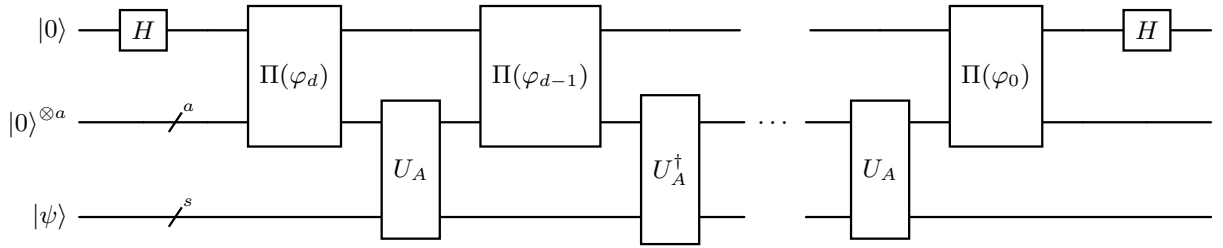

	\subsection{Robust Oblivious Amplitude Amplification (ROAA)}
	\label{subsec:ROAA}
	
	Amplitude amplification is essential to ensure that measuring the ancilla qubits in \(\ket{0^a}\) occurs with high probability, thus extracting the correctly transformed state from a block-encoded operator. While the idea of amplitude amplification dates back to Grover’s pioneering search algorithm \cite{grover1996fast}, standard approaches can be cumbersome when dealing with block encodings and ancillary subspaces. In these cases, \emph{robust oblivious amplitude amplification} (ROAA) \cite{Berry_Childs_Cleve_Kothari_Somma_2014} is frequently preferred as it requires fewer assumptions about re-preparing states and does not incur significant overhead from repeated projective measurements. Indeed, unlike standard amplitude amplification which often involves explicit measurement or reinitialization of the ancillas, ROAA provides an “oblivious” mechanism that preserves the block-encoded structure throughout the amplification steps, minimizing disruption to the overall circuit while still boosting success probabilities to near unity \cite{Martyn2021,KothariThesis}. 
	
	\begin{figure*}[ht!]
		\centering
		\subfloat[]
		{
			\begin{adjustbox}{width=0.2\textwidth}
				\begin{quantikz}
					\lstick[5]{$x$}& \gate{X} & \ctrl{2}	& \gate{X}	& \\
					& \gate{X} & \control{}	& \gate{X}	& \\
					& \setwiretype{n} \vdots & \vdots     & \vdots    & \\
					& \setwiretype{n}  &      &     & \\
					& \gate{X} & \gate{Z}\wire[u][1]{q} & \gate{X} &
				\end{quantikz}
			\end{adjustbox}
		}
		\subfloat[]
		{
			\begin{adjustbox}{width=0.80\textwidth}
				\begin{quantikz}[slice style = blue]
					\lstick{$\ket{0}^{\otimes a}$} &\qwbundle{a} & \slice{$\ket{0}^{\otimes a}\ket{v}$} &\gate[2]{U_A} 	& \gate{R_{\psi_g} = R_0(a)}\gategroup[2, steps=6, style={dashed, rounded corners, fill = blue!20, inner xsep = 2pt}, background, label style={label position = above, anchor = north, yshift=0.4cm}]{\textcolor{blue}{$W$}} 	&\gate[2]{U_A^\dagger} 	& \gategroup[2, steps=3, style={dashed, rounded corners, fill = red!20, inner sep = 1pt}, background, label style={label position = above, anchor = north, yshift=0.43cm}]{\textcolor{red}{$R_{\psi_0}$}}						& \gate[2]{R_{0}(a+n)} & 			& \gate[2]{U_A} & \meter{} \rstick[brackets=none]{$\ket{0}^{\otimes a}$} \\
					\lstick{$\ket{0}^{\otimes n}$} &\qwbundle{n} & \gate{S_n} 	& 				& 								& 						& \gate{S_n^\dagger} 	& 			 		&\gate{S_n} & 				& \rstick[brackets=none]{$\frac{\tilde{A} \ket{v}}{\norm{\tilde{A} \ket{v}}_2}$}
				\end{quantikz}
			\end{adjustbox}
		}
		\caption{(a) Circuit realizations of a reflection operator  $R_0(m)$ and the Grover iteration operator $W$. The reflection operator $R_0(m)$ applies a phase of $-1$ only if the input is $\ket{0^x}$. It decomposes into $X$ and multi-controlled $Z$ gates. (b)  The Grover iteration operator $W$ from Definition~\ref{def:AA_W_iteration} uses two reflections, $R_{\psi_0}$ and $R_{\psi_g}$, along with the block-encoding $U_A$ and its inverse $U_A^\dagger$. Iterating $W$ amplifies the amplitude on $\ket{0^a}\otimes \tilde{A}\ket{v}$.} 
		\label{fig:AA_amplitude}
	\end{figure*}
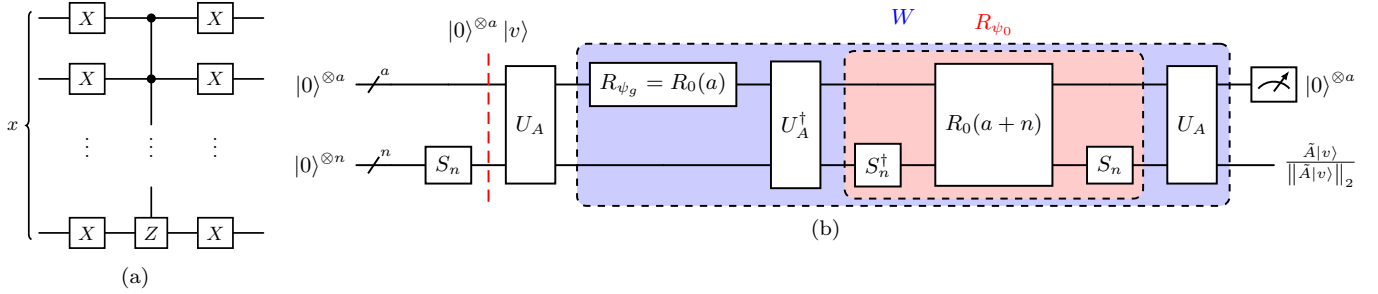
	
	\vspace{3pt}
	\paragraph*{Amplitude amplification in a block-encoding.} Suppose
	\[
	\ket{\Psi_0}
	\;=\;
	\bigl(I_a\otimes S_n\bigr)
	\,\bigl(\ket{0^a}\otimes\ket{0^n}\bigr),
	\]
	where $S_n$ is a state-preparation unitary mapping $\ket{0^n}\mapsto \ket{v}$. Applying a \((\alpha,a,0)\)-block-encoding $U_A$ of $A$ onto $\ket{\Psi_0}$ yields (up to normalization)
	\begin{equation}\label{eq:AA_BE}
		U_A\,\ket{\Psi_0}
		\;=\;
		\ket{0^a}\otimes \tilde{A}\ket{v}
		\;\;+\;\;
		\sqrt{1-\norm{\tilde{A}\ket{v}}_2^2}\,\ket{\Psi_{\text{bad}}},
	\end{equation}
	where \(\tilde{A}:=A/\alpha\). Observing that projecting the ancillas onto $\ket{0^a}$ collapses the state to $\tilde{A}\ket{v}$, one can interpret $\ket{0^a}\otimes\tilde{A}\ket{v}$ as the ``good'' component and the orthogonal remainder as ``bad.’’  
	
	\begin{definition}[Grover Iteration]\label{def:AA_W_iteration}
		Define the Grover iteration operator
		\begin{equation}\label{eq:AA_W}
			W
			\;\coloneqq\;
			-\,U_A
			\;R_{\Psi_0}\;
			U_A^\dagger
			\;R_{\Psi_g},
		\end{equation}
		where
		\begin{align}
			R_{\Psi_0}&\;=\;I-\;2\,\dyad{\Psi_0}{\Psi_0}
			\;=\;(I_a \otimes S_n)\,R_0(a+n)\,(I_a \otimes S_n^\dagger),\\[6pt]
			R_{\Psi_g}&\;=\;I-\;2\,\dyad{\Psi_g}{\Psi_g}
			\;=\;I_{a+n}-\,2\,\bigl(\dyad{0^a}{0^a}\otimes I_n\bigr).
		\end{align}
		Here, a reflection operator $R_0(m)\coloneqq I_m-2\,\dyad{0^m}{0^m}$ flips the sign only of $\ket{0^m}$, and $\ket{\Psi_g}$ denotes the ``good’’ subspace $\ket{0^a}\otimes \tilde{A}\ket{v}$ in this context.
	\end{definition}
	
	The circuit realization of the reflection operator $R_0(x)$ and of the Grover iteration operator $W$ are detailed in Fig.~\ref{fig:AA_amplitude}. Next, iterating on $W$ systematically amplifies the amplitude on $\ket{0^a}\otimes \tilde{A}\ket{v}$, driving the success probability toward near unity. Specifically, while the success probability of measuring \(\ket{0^a}\) on the ancilla wire before amplification is \(p = \|\tilde{A}\ket{v}\|_2^2\), this probability can be boosted close to one by applying \(W\) \(\mathcal{O}(1/\sqrt{p})\) times (see \cite{Berry_Childs_Cleve_Kothari_Somma_2014,Martyn2023}).
	
	Thus, robust oblivious amplitude amplification (ROAA) provides the final piece for reliably extracting the “good” subspace after a polynomial transformation via QSVT. In practice, once the block-encoding $U_A$ for $A$ is constructed, we append amplitude amplification steps until the ancilla-measurement probability surpasses a desired threshold. Together, these block-encoding, QSVT, and amplitude amplification routines form the backbone of our quantum algorithm for simulating coupled oscillators.

	
	\section{Theoretical Results}
	\label{sec:result_theory}
	
	In this section, we present our main results on the explicit quantum circuit for simulating coupled oscillator chains. Our systems of interest include both periodic (\emph{closed-chained}) and open (\emph{open-chained}) boundary conditions, with masses and spring constants that may be either uniform or non-uniform. We begin by constructing a block encoding circuit of the key matrix \(B\) which crucially encodes the mass-spring network coupling structure (see ~\eqref{eq:CO_weightB}). From this, we assemble a block encoding of the Hamiltonian \(H\) ~\eqref{eq:CO_H}, using the previously obtained block-encoded \(B\). Finally, by composing QSVT-based circuits (implementing polynomials like \(\cos(\mathcal{H}t)\) and \(-\,i\,\sin(\mathcal{H}t)\)) via linear combination of unitaries (LCU), we realize a block encoding circuit of the time-evolution operator \(e^{-\,i\, \mathcal{H}\, t}\). To amplify the probability of measuring the ancilla qubits in the correct (projective) outcome, we apply robust oblivious amplitude amplification, yielding a high-fidelity target time-evolved state  at the circuit’s output.
	
	\subsection{Block Encoding Circuit of \boldmath\(B\)}
	\label{subsec:BE_B}
	
	\subsubsection{Uniform Coupled Oscillators with \boldmath\(N = 2^n\) Oscillators}
	\label{sss:BE_B_unit}
	
	We first consider a simple scenario in which each oscillator is connected only to its nearest neighbor, and all masses and spring constants are set to 1 (\(m_i = k_{(i,i+1)}=1\)). Let \(N=2^n\) be the total number of oscillators, with \(n \in \mathbb{N}\). Depending on the boundary conditions, we introduce
	\begin{equation}
		\label{eq:A_equiMK}
		\mathcal{A}
		\;=\;
		\begin{pmatrix}
			\mathfrak{C}+1 & -1 & 0 & \cdots & 0 & -\,\mathfrak{C}\\
			-1 & 2 & -1 &  \ddots & 0 & 0 \\
			0 & -1 & 2 &  \ddots & \vdots & 0 \\
			\vdots &  \ddots &  \ddots & \ddots & -1 & \vdots \\
			0 & 0 & \hdots & -1 & 2 & -1 \\
			-\mathfrak{C} & 0 & 0 & \hdots & -1 & \mathfrak{C}+1
		\end{pmatrix}_{\!N\times N}.
	\end{equation}
	where
	\begin{equation}
		\label{eq:c_withCyC}
		\mathfrak{C} \;=\;
		\begin{cases}
			0, & \text{open-chained},\\
			1, & \text{closed-chained}.
		\end{cases}
	\end{equation}
	One easily verifies that \(\mathcal{A}\) is a Laplacian matrix that can be factorized as \(\bar{B}\,\bar{B}^\dagger\), where \(\bar{B}\) is derived from $B$ in ~\eqref{eq:CO_weightB} but specialized to unit masses and springs (we add “\(\bar{\phantom{B}}\)” to denote uniform systems). Because \(\sqrt{M} = I\) and \(\sqrt{W}=I\) in this uniform setting, one obtains 
	\begin{equation}
		\label{eq:B_equiMK}
		\bar{B}
		\;=\;
		\begin{pmatrix}
			1 & 0 & 0 & \cdots & -\,\mathfrak{C}\\
			-1 & 1 & 0 & \cdots & 0\\
			0 & -1 & 1 & \cdots & \vdots\\
			\vdots & \vdots & \vdots & \ddots & 0\\
			0 & 0 & 0 & \cdots & \mathfrak{C}
		\end{pmatrix}_{N\times N}.
	\end{equation}
	We now construct block encodings of \(\bar{B}\) and analyze their gate complexities in both boundary conditions.
	
	\paragraph*{1.1 Closed-chained case.}
	Denote by \(\bar{B}_c(N)\) the closed-chain matrix \(\bar{B}\) for \(N\) oscillators. Referring to ~\eqref{eq:B_equiMK}, its nonzero elements appear only on the diagonal and certain off-diagonal positions, so we can decompose
	\begin{equation}
		\label{eq:decom_cy_B}
		\bar{B}_c(N)
		\;=\;
		I_N \;-\; L_N,
	\end{equation}
	where the \emph{L}-shift matrix \(L_N\) is defined as
	\[
	L_N
	\;=\;
	\begin{pmatrix}
		0 & 0 & \cdots & 0 & 1\\
		1 & 0 & \cdots & 0 & 0\\
		0 & 1 & \cdots & \vdots & \vdots\\
		\vdots & \vdots & \ddots & 0 & 0\\
		0 & 0 & \cdots & 1 & 0
	\end{pmatrix}_{N\times N}.
	\]
	In block-encoding language, \(L_N\) admits a trivial \((1,0,0)\)-block-encoding, so it requires no additional ancillas. A quantum circuit for the \emph{L}-shift matrix \(L_{2^n}\) can be built via a series of multi-controlled NOT gates (C$^n$NOTs). Explicitly,
	\[
	L_{2^n}
	\;=\;
	\prod_{j=0}^{n-1}
	\mathrm{C}^{n-j-1}\mathrm{NOT}_{\{m \,\mid\,m>j+1,m\in[n]\}}^{\,j+1},
	\]
	where each term is a multiple-controlled NOT with $n-j$ control node(s) acting on a set of controlled qubit(s) $\{m|m>j,m\in [n]\}$ and targeting the $j$th qubit. For $j = n-1$, the term becomes a NOT gate. A \emph{$L$}-shift gate $L_{2^n}$ is simply made up of multi-controlled NOT gate where the target qubit starts from the topmost qubit and decreases by one for each gate. Figure~\ref{circ:L_n} shows an explicit circuit of $L_{2^n}$.
	
	\begin{figure}[h]
		\centering
		\begin{quantikz}
			\lstick[4]{$n$}& \targ{}					& 						& \midstick[2, brackets=none]{$\cdots$}		&					&  			&\\
			& \ctrl{1} \wire[u][1]{q}	& \targ{} \wire[d][1]{q}&											& 					&  			&\\
			&\setwiretype{n}  	& 				& \vdots									& 			&    & \\ 	 
			& \control{} \wire[u][1]{q}	& \ctrl{-1}				& 											&					& \gate{X}  &
		\end{quantikz}
		\caption{A quantum circuit for implementing the \emph{L}-shift operator $L_{2^n}$. It consists of a sequence of multi-controlled NOT gates plus one $X$ gate.}
		\label{circ:L_n}
	\end{figure}
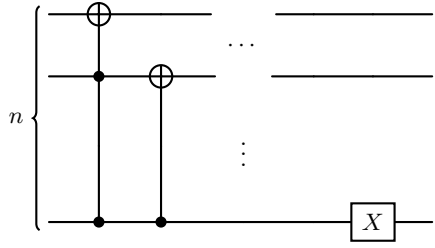
	
	By linear combination of block-encoded matrices (LCU), one obtains a block encoding of $\bar{B}_c(N)$ using the decomposition ~\eqref{eq:decom_cy_B}.
	The resulting circuit is
	\begin{equation}
		U_{\bar{B}_{c}(N)}
		\;=\;
		\bigl(XH \otimes I_N\bigr)
		\;\bigl(C L_N\bigr)
		\;\bigl(H \otimes I_N\bigr),
	\end{equation}
	with $C L_N$ denoting a controlled-$L_N$ gate. Figure~\ref{circ:notgen_BE_B}(a) illustrates this circuit. The total qubit count for building the block encoding of \(\bar{B}_c(N)\) is \(1 + \log_2 N\).
	
	\paragraph*{Gate complexity of \boldmath\(\overline{B}_c\).}
	Decomposing multi-controlled gates into single- and two-qubit elementary gates \cite{Barenco_1995} reveals that \(\bar{B}_c(N)\) has gate complexity scaling as
	\begin{equation}
		\label{eq:gateCom_Bcy}
		\bar{G}_c(N)
		\;=\;
		\mathcal{O}\bigl(n^2\bigr)
		\;=\;
		\mathcal{O}\bigl(\log_2^2 N\bigr).
	\end{equation}
	Table~\ref{tab1} lists the types and counts of gates.
	
	\begin{table}[h]
		\centering
		\caption{{\bf Gate counts in constructing the block encoding of \(\bar{B}_c(N)\).} Here, the V gate is the square root of the Pauli-\(X\) gate, which can be directly implemented in Qiskit as the \(\texttt{sx}\) instruction.}
		\label{tab1}
		\setlength{\tabcolsep}{3pt}
		\begin{tabular}{|c|c|}
			\hline
			\textbf{Gate type} & \textbf{Number of calls}\\
			\hline
			$H$	& 2 \\
			\hline
			$X$ & 1\\
			\hline
			CNOT & $31 \;+\; 8 \cdot \displaystyle \sum_{i=5}^{n} 4(i-2) \sim \mathcal{O}(n^2)$\\
			\hline
			CV	& 22\\
			\hline
			V & $8 \;+\; 8 \cdot \displaystyle \sum_{i=5}^{n} 4(i-2) \sim \mathcal{O}(n^2)$\\
			\hline
		\end{tabular}
	\end{table}
	
	\paragraph*{1.2 Open-chained case.}
	For the open-chained scenario, the matrix $\bar{B}$ has zero in its last column. Denote this matrix as $\bar{B}_o(N)$; it splits into
	
	\begin{align*} 
		\bar{B}_o(N)  &\coloneq I'_N - L'_N \\
		&= \mqty(1 & 0 & \cdots & 0 & 0 \\ 0 & 1 & \cdots & 0 & 0 \\ \vdots& \vdots& \ddots & \vdots& \vdots \\ 0& 0& \cdots & 1 & 0\\ 0 & 0& \cdots &0 &0)_{N \times N} - \mqty(0 & 0 & \cdots & 0 & 0 \\ 1 & 0 & \cdots & 0 & 0 \\ 0 & 1 & \cdots & \vdots & \vdots \\ \vdots &\vdots &\ddots & 0 & 0\\ 0 & 0 & \cdots & 1 & 0)_{N \times N},  \numberthis \label{eq:decom_ncy_B}
	\end{align*}
	
	\noindent where the \emph{prime} notation in $I'_N$ and $L'_N$ indicates that the last column 
	of the original $I_N$ and $L_N$ is set to zero. Both $I'_N$ and $L'_N$ each have straightforward \((1,1,0)\)-block-encodings, as shown in Fig.~\ref{circ:ncy_A01}. Combining them by LCU yields a \((2,2,0)\)-block-encoding of \(\bar{B}_o(N)\): 
	\begin{align}
		U_{\bar{B}_o(N)} & = (X\otimes I_{2N})(H \otimes I_{2N})(CL_{2N})\cdots \\
		&\cdots(X\otimes I_{2N})(\mathrm{C}^N \mathrm{NOT}^1_{m\in[\log_2 N]})(X\otimes I_{2N})(H\otimes I_{2N}).\nonumber
	\end{align}
	The circuit is depicted in Fig.~\ref{circ:notgen_BE_B}(b). 
	
	\begin{figure}[!h]
		\centering
		\subfloat[]{
			\begin{adjustbox}{width=0.195\textwidth}
				\begin{quantikz}
					\lstick{$\ket{0}$} 		& \gate[4]{U_{I'_N}} 	& \midstick[4, brackets=none]{=} 	& \targ{} \wire[d][2]{q}	&\\
					\lstick[3]{$n$}			& 					&									& \control{}				&\\
					\setwiretype{n}			&					& 									& \vdots						&\\
					&					&									& \control{} \wire[u][1]{q}	&
				\end{quantikz}
			\end{adjustbox}
		}
		\subfloat[]{
			\begin{adjustbox}{width=0.27\textwidth}
				\begin{quantikz}
					\lstick{$\ket{0}$} 		& \gate[4]{U_{L'_N}} 	& \midstick[4, brackets=none]{=} 	&\gate[4]{L_{2N}}		&\\
					\lstick[3]{$n$}			& 					&									&						&\\
					\setwiretype{n}			&					& 									&						&\\
					&					&									&						&
				\end{quantikz}
			\end{adjustbox}
		}
		\caption{(a)  A \((1,1,0)\)-block-encoding of $I'_N$, realized by a multi-controlled NOT gate. 
			(b) A \((1,1,0)\)-block-encoding of $L'_N$, realized by an $L$-shift gate on one extra qubit.} 
		\label{circ:ncy_A01}
	\end{figure}
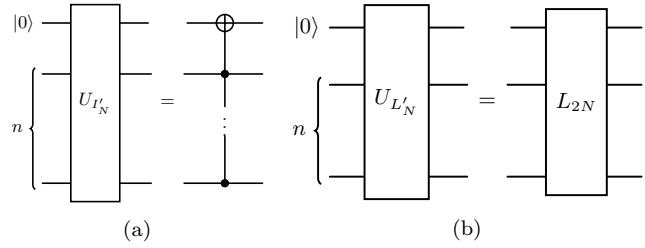
	
	\paragraph*{Gate complexity of \boldmath\(\overline{B}_o(N)\).}
	Similar to the closed-chained case, the open-chain construction adds one extra ancilla qubit 
	and a multi-controlled NOT gate. Concretely, for $\bar{B}_o(N)$, the circuit involves 
	2\,C$^{n+1}$NOT gates, 2\,Pauli-$X$ gates, and the base $\bar{B}_{o}(N)$ block encoding. 
	Since the complexity for closed-chained \(\bar{B}_c(N)\) is 
	$\mathcal{O}(\log_2^2 N)$, this open-chained variant remains at the same order: 
	\begin{equation}
		\bar{G}_c(N) = \mathcal{O}(\log_2^2 N).
	\end{equation}

	\paragraph*{Conclusion for uniform \boldmath\(\overline{B}(N)\).}
	In both boundary conditions, we have explicit block-encoding circuits of $\bar{B}(N)$ with gate complexity $\mathcal{O}\bigl(\log_2^2 N\bigr)$. Next, we will utilize these block encoding circuits to obtain the block encodings of the more general matrix $B$.
	
	\begin{figure*}[h!]
		\centering
		\subfloat[]{%
			\begin{adjustbox}{width=2.0in}
				\begin{quantikz}
					\lstick{$\ket{0}$} &		&\gate{H}		& \ctrl{1}	&\gate{H}	&\gate{X} &\\
					& \qwbundle{n}&			& \gate{L_N}	&		&&
				\end{quantikz}
			\end{adjustbox}
		}
		\quad
		\subfloat[]{%
			\begin{adjustbox}{width=3.75in}
				\begin{quantikz}
					\lstick{$\ket{0}$}	&					&\gate{H}		& \ctrl[open]{1}	&\ctrl{1}	&\gate{H}	&\gate{X} &\\
					& \qwbundle{n+1}	&				& \gate{U_{I'_N}}		&\gate{U_{L'_N}}	&			&			&
				\end{quantikz}
				=
				\begin{quantikz}
					\lstick[2]{$\ket{0}^{\otimes 2}$} 			&\gate{H}			& \ctrl[open]{1}		&\ctrl{1}			&\gate{H}		&\gate{X} 		&\\
					&					&\targ{} \wire[d][2]{q}\gategroup[4, steps=1, style={dashed, rounded corners, fill = blue!10, inner xsep = 2pt}, background, label style={label position = below, anchor = north, yshift=-0.4cm}]{$U_{I'_N}$}	&\gate[4]{L_{2N}}\wire[d][2]{q}\gategroup[4, steps=1, style={dashed, rounded corners, fill = blue!10, inner xsep = 2pt}, background, label style={label position = below, anchor = north, yshift=-0.4cm}]{$U_{L'_N}$} 	&			&				&\\
					\lstick[3]{$n$}								&					&\control{}				&					&				&				&\\
					\setwiretype{n}								&					&\vdots					&					&				&				&\\
					&					&\control{} \wire[u][1]{q}&					&				&				&
				\end{quantikz}
			\end{adjustbox}
		}
		\caption{(a) A \((2,1,0)\)-block-encoding of the closed-chained operator \(B_c(N)\), needing $\log_2 N + 1$ qubits; its gate count scales as $\mathcal{O}(\log_2^2 N)$. 
			(b) A \((2,2,0)\)-block-encoding of the open-chained operator \(B_o(N)\), needing $\log_2 N + 2$ qubits and also $\mathcal{O}(\log^2 N)$ gates.} 
		\label{circ:notgen_BE_B}
	\end{figure*}
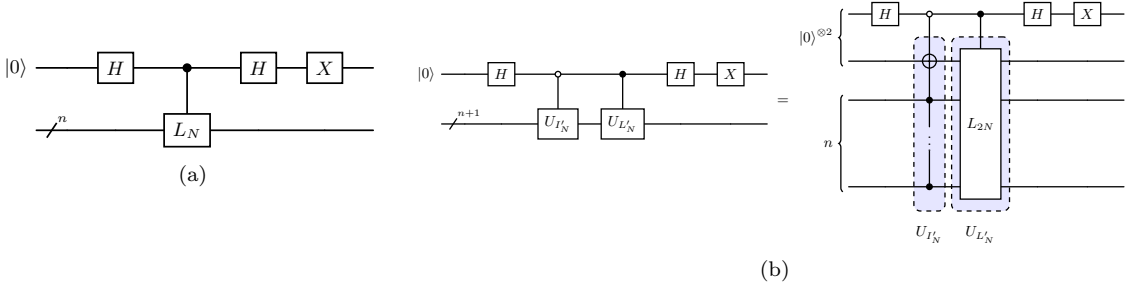
	
	\subsubsection{Generalized Coupled Oscillators with \boldmath\(N=2^n\)}
	\label{sss:genMK_CCO}
	
	We now move beyond the uniform-mass-and-spring case, allowing each mass \(m_j\) and spring constant \(k_{(j,j\pm1)}\) to vary throughout the chain. However, we still assume \(N=2^n\) for simplicity, so that index registers naturally fit a \(\log_2 N\)-qubit format. As before, we concentrate on one-dimensional chains with either open or periodic boundary conditions (cf.\ \eqref{eq:c_withCyC}). 
	
	\paragraph*{Circuit construction for block-encoding of \boldmath\(B\).}
	Recall from ~\eqref{eq:CO_weightB} that the matrix \(B\) for the general system is
	\[
	B
	\;=\;
	\sqrt{M}^{-1}\,\Phi_F\,\sqrt{W},
	\]
	where \(\Phi_F\) is the incidence matrix for the graph of the spring network, \(W\) is a diagonal matrix of spring constants~\eqref{eq: W_matrix}, and \(M\) is the diagonal mass matrix~\eqref{eq:mass_matrix}. For each of these diagonal operators, one constructs a \(\bigl(1,1,0\bigr)\)-block-encoding by applying bitstring-controlled $R_y$ rotations. Specifically, if
	\[
	D
	\;=\;
	\mathrm{diag}\bigl(d_0,d_1,\dots,d_{N-1}\bigr),
	\]
	then
	\begin{equation}
		U_D
		\;=\;
		\prod_{i=0}^{N-1}
		\Bigl[
		R_y\bigl(\theta_i\bigr)\otimes \dyad{i}{i}
		\Bigr],
		\quad
		\theta_i
		\;=\;
		2\,\arccos \bigl(d_i\bigr).
	\end{equation}
	Note that each $d_i$ must lie within $[-1,1]$ to avoid complex angles. Hence, if we wish 
	to set $d_i = \sqrt{k_i}$ or $d_i = \sqrt{m_i}^{-1}$, one typically enforces
	\[
	\sqrt{k_{i,i \pm 1}} \;\le\; 1 
	\quad\text{and}\quad 
	\sqrt{m_i}^{-1} \;\le\; 1,
	\]
	which is equivalent to $0 \le k_{i, i \pm 1} \le 1$ and $m_i \ge 1$. 
	As physical systems may have spring constants or masses outside these 
	ranges. One addresses this by \emph{rescaling} all $\{k_{i,i\pm 1}\}$ by $k_{\max}$  and all $\{m_i\}$ by $1/m_{\min}$ (or another suitable factor), 
	bringing each diagonal entry into $[-1,1]$. 
	However, such rescaling increases the overall 
	sub-normalization constant in the block-encoded operator, potentially introducing 
	additional overhead in amplitude amplification or query complexity. Consequently, 
	while not mandatory \emph{ab initio}, the condition $k_{i,i \pm 1} \in (0,1]$ and 
	$m_i \in [1,\infty)$ can be obtained after an appropriate rescaling 
	to ensure real-valued $R_y$ rotations in block encodings.
	
	Next, recall that in the uniform case where each mass and spring constants equals 1, the 
	incidence matrix $\Phi_F$ coincides with the matrix $\bar{B}$. Consequently, the 
	same block-encoding circuits used for $\bar{B}_c$ and $\bar{B}_o$ 
	(cf.\ Fig.~\ref{circ:notgen_BE_B}) can be adapted here to encode the corresponding $\Phi_F^c$ and $\Phi_F^o$ respectively. Finally, we form the product
	$B_{c/o}
	\;=\;
	\sqrt{M}^{-1}\;\Phi_F^{c/o}\;\sqrt{W}
	$
	by composing each block-encoded factor via the standard block-encoding composition rules, which provide a new block-encoding of the product (Appendix \ref{appSSec:BE_prod}). 
	
	\begin{figure*}[hpt!]
		\centering
		\subfloat[\label{circ:BE_diag}]{%
			\begin{adjustbox}{width=3.25in}
				\begin{quantikz}
					\lstick{$\ket{0}$}		&\gate{R_y(\theta_0)} \wire[d][1]{q}		&\gate{R_y(\theta_1)} \wire[d][1]{q} 	&\gate{R_y(\theta_2)} \wire[d][1]{q} & &\midstick[wires=5,brackets=none]{$\cdots$} &\gate{R_y(\theta_{N-1})} \wire[d][1]{q}& \\
					\lstick[4]{$n$}	&\ctrl[open]{2}								&\ctrl[open]{2}			&\ctrl[open]{2}	 & & &\ctrl{2} &\\
					& \setwiretype{n}  				&							& 		&\midstick[brackets=none]{$\vdots$} & & &\\
					&\ctrl[open]{1}							&\ctrl[open]{1}				&\ctrl{1} 		& & &\ctrl{1} &\\
					&\control[open]{}						&\control{}					&\control[open]{} & & &\control{} &
				\end{quantikz}
			\end{adjustbox}
		}
		\hfill
		\subfloat[\label{circ:BE_genB_cy}]{%
			\begin{adjustbox}{width=1.85in}
				\begin{quantikz}
					\lstick{$\ket{0}_{\sqrt{M}^{-1}}$}	&					&&&				& 				&\gate[4]{U_{\sqrt{M}^{-1}}}	&\\
					\lstick{$\ket{0}_{\Phi_F}$}	&					&&&				&\gate[3]{U_{\Phi_F}}	&				&\\
					\lstick{$\ket{0}_{\sqrt{W}}$}	&					&&&\gate[2]{U_{\sqrt{W}}}	&				&				&\\
					&\qwbundle{n}&&&				&				&				&\\	
				\end{quantikz}
			\end{adjustbox}
		}
		\hfill
		\subfloat[\label{circ:BE_genB_ncy}]{%
			\begin{adjustbox}{width=1.75in}
				\begin{quantikz}
					\lstick{$\ket{0}_{\sqrt{M}^{-1}}$}	&					&&&				& 				&\gate[5]{U_{\sqrt{M}^{-1}}}	&\\
					\lstick[2]{$\ket{0}_{\Phi_F}$}	&					&&&				&\gate[4]{U_{\Phi_F}}	&				&\\
					&					&&&				&				&				&\\
					\lstick{$\ket{0}_{\sqrt{W}}$}	&					&&&\gate[2]{U_{\sqrt{W}}}	&				&				&\\	
					&\qwbundle{n}&&&				&				&				&
				\end{quantikz}
			\end{adjustbox}	
		}
		\caption{(a) Circuit construction of a block encoding for a diagonal matrix
			\(\mathrm{diag}(d_0, d_1, \dots, d_{N-1})\), using a sequence of bit-string controlled 
			\(R_y\bigl(2\arccos(d_i)\bigr)\) gates. This approach extends directly to realize 
			a generalized \(B_{c,o}(N)\) by combining block-encoded \(\Phi_F\) and diagonal factors 
			via multiplication. Panels (b) and (c) depict the \((1,1,0)\)- and \((1,3,0)\)-block-encodings 
			of \(B_c(N)\) and \(B_o(N)\), respectively, each of which has gate complexity 
			\(\mathcal{O}\bigl(2\,N\,\log_2 N + \log_2^2N\bigr)\) . In (b) and (c), $\ket{0}_M, \ket{0}_{\Phi}$, and $\ket{0}_W$ denote the ancilla wires which extended independently according to the block encoding of $\sqrt{M}^{-1},\Phi_F^{c/o}$, and $\sqrt{W}$, respectively.}
	\end{figure*}
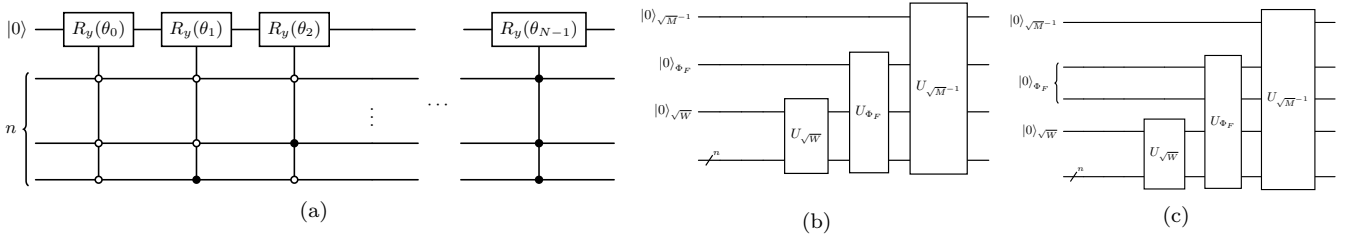
	
	\paragraph*{Closed-chained vs.\ open-chained gate complexity.}
	We now state the gate complexity analysis for both boundary conditions. 
	Recall that block encoding each diagonal matrix $\bigl(\sqrt{W}$ or $\sqrt{M}^{-1}\bigr)$ 
	requires $N$ multi-controlled rotations, one per diagonal entry. Decomposing each 
	multi-controlled rotation into single- and two-qubit gates typically costs 
	$\mathcal{O}\bigl(\log_2^r N\bigr)$, with $r=1$ or $2$ depending on the chosen 
	construction \cite{Barenco_1995}. Hence, encoding both diagonal operators contributes 
	$\mathcal{O}\bigl(2\,N\,\log_2^r N\bigr)$ in total. Meanwhile, the incidence-like matrix 
	$\Phi_F$ (which coincides with $\bar{B}$ in the uniform setting) incurs 
	$\mathcal{O}\bigl(\log_2^2 N\bigr)$ gates. Combining these yields 
	\[
	\mathcal{O}\!\bigl(2\,N\,\log_2^r N \;+\;\log_2^2 N\bigr)
	\]
	gates to block-encode $B$. 
	
	For closed-chained systems, this procedure realizes a $(1,3,0)$--block-encoding circuit 
	of $B_c$, illustrated in Fig.~\ref{circ:BE_genB_cy}. Likewise, open-chained systems 
	obtain a $(1,4,0)$--block-encoding circuit of $B_o$ (Fig.~\ref{circ:BE_genB_ncy}) 
	by adding an extra ancilla qubit for the boundary condition. In \emph{best-case} 
	(\emph{i.e.}, $r=1$ (linear overhead)) multi-controlled decompositions, both scenarios achieve
	\begin{equation}
		\label{eq:cost_generalBcBo_revised}
		G_c(N)
		\;=\;
		G_o(N)
		\;=\;
		\mathcal{O}\Bigl(2\,N\,\log_2 N + \log_2^2 N\Bigr).
	\end{equation}
	
	\subsubsection{Generalized Coupled Oscillators with \boldmath\(N \neq 2^n\)}
	\label{sss:genMKnoOsc_CCO_short}
	
	When the total number of oscillators \(N\) is not a power of two, we can embed the system into a larger dimension of size \(2^{\tilde{n}}\), where \(\tilde{n} \coloneqq \lceil\log_{2}N\rceil\). Only the incidence matrix \(\Phi_F\) requires explicit zero-padding (denoted \(\tilde{\Phi}_F\)), since any extra diagonal entries in the weight or mass matrices are automatically nullified. As before, \(\tilde{\Phi}_F\) decomposes into a shifted identity minus an $L$-shift term, but now includes rows and columns of zeros to match the \(2^{\tilde{n}}\)-dimension.
	
	Constructing these enlarged block encodings effectively amounts to moving any unwanted ones out of the top-left block, a process that adds only \(\mathcal{O}(\tilde{n})\) gates and one extra ancilla qubit. Consequently, the overall gate complexity for building $B$ in this non–\(2^n\) scenario remains the same order as before. Specifically, both closed- and open-chained systems achieve
	\begin{align*}
		\tilde{G}_c(N)
		=
		\tilde{G}_o(N)
		&=
		\mathcal{O}\Bigl( 
		2\tilde{N}\log_2 \tilde{N}+\log_2^2 \tilde{N} +\log_2 \tilde{N}\Bigr)\\
		&=\mathcal{O}\Bigl( 
		2\tilde{N}\log_2 \tilde{N}+\log_2^2 \tilde{N} \Bigr)
	\end{align*}
	where \(\tilde{N} \coloneqq 2^{\tilde{n}}\). We detail the full procedure along with representative circuit examples in the Appendix \ref{app:genMKnoOsc_CCO}.

	\subsection{Block Encoding Circuit of \boldmath\(\mathcal{H}\)}
	\label{subsec:BE_H}
	
	\begin{figure}[ht!]
		\centering
		\begin{quantikz}
			&\gate{H}			&\ctrl[open]{1}		&\ctrl{1}	&\gate{H}&\gate{R_y(2\pi)}& \\
			&\qwbundle{a+1}		&\gate[2]{U_B^{(01)}}		&\gate[2]{{U_B^{(10)}}}&&&\\
			&\qwbundle{n+1}		&					&						&&&
		\end{quantikz}
		\begin{adjustbox}{width=0.49\textwidth}
			=\begin{quantikz}
				\lstick[4]{$\ket{0}^{\otimes a_B+2}$}			&\gate{H}		&\gate{X}			&\ctrl{2}		&\ctrl{2}		&\ctrl{1}		&\ctrl{3}		&\ctrl{2}		&\gate{X}		&\ctrl{2}		&\ctrl{1}		&\ctrl{2}		&\ctrl{3}		&\ctrl{2}	& \gate{H}&\gate{R_y(2\pi)}&\\
				&				&					&\gategroup[5, steps=5, style={dashed, rounded corners, fill = blue!10, inner xsep = 2pt}, background, label style={label position = below, anchor = north, yshift=-0.4cm}]{\textcolor{blue}{$U_B^{(01)}$}}				&\gategroup[2, steps=2, style={dashed, rounded corners, fill = red!10, inner xsep = 2pt}, background, label style={label position = below, anchor = north, yshift=-0.2cm}]{\textcolor{red}{$U_{01}$}}				&\swap{1}		&				&				&				&\gategroup[5, steps=5, style={dashed, rounded corners, fill = blue!10, inner xsep = 2pt}, background, label style={label position = below, anchor = north, yshift=-0.4cm}]{\textcolor{blue}{$U_B^{(10)}$}}				&\swap{1}\gategroup[2, steps=2, style={dashed, rounded corners, fill = red!10, inner xsep = 2pt}, background, label style={label position = below, anchor = north, yshift=-0.2cm}]{\textcolor{red}{$U_{10}$}}		&				&				&			&&&\\
				&				&					&\swap{2}		&\targ{}		&\targX{}		&				&\swap{2}		&				&\swap{2}		&\targX{}		&\targ{}		&				&\swap{2}	&&&\\
				&\qwbundle{a_B-1}	&					&				&				&				&\gate[3]{U_B}	&				&				&				&				&				&\gate[3]{U_B^\dagger}&		&&&\\
				\lstick[2]{$\ket{\psi}$}					&				&					&\targX{}		&				&				&				&\targX{}		&				&\targX{}		&				&				&				&\targX{}	&&&\\
				&\qwbundle{n}	&					&				&				&				&				&				&				&				&				&				&				&			&&&
			\end{quantikz}
		\end{adjustbox}
		\caption{{\bf An explicit \(\bigl(2\,\alpha_B,\,a_B+2,\,0\bigr)\)-block-encoding circuit of the Hamiltonian \(\mathcal{H}\), constructed from a \(\bigl(\alpha_B,\,a_B,\,0\bigr)\)-block-encoding of \(B\).} The circuit consists of two block-encodings of the operator 
			\(\dyad{0}{1}\otimes B\),  and its Hermitian-conjugate counterpart \(\ket{1}\!\!\bra{0}\otimes B^\dagger\).The circuits $U_{01}$ and $U_{10}$ represent a block encoding of $\dyad{0}{1}$ and $\dyad{1}{0}$, respectively. These can be done by using SWAP and NOT gate. Finally, an \(R_y(2\pi)\) rotation applies an overall factor of \(-1\) to match \(\mathcal{H} = -\bigl(\ket{0}\!\!\bra{1}\otimes B + \ket{1}\!\!\bra{0}\otimes B^\dagger\bigr)\).} 
		\label{circ:BE_H}
	\end{figure}
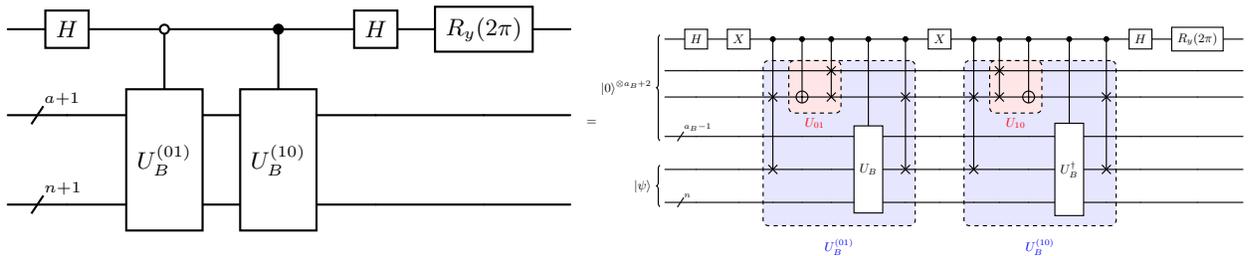
	
	After the block-encoding of \(B\) is obtained, we next build a block-encoding of the Hamiltonian \(\mathcal{H}\) in ~\eqref{eq:CO_H}. Recall 
	\[
	\mathcal{H} 
	\;=\;
	-\,\bigl(\dyad{0}{1}\otimes B\bigr)\;-\;\bigl(\dyad{1}{0}\otimes B^\dagger\bigr).
	\]
	We employ a linear-combination-of-unitaries (LCU) approach, combining:
	\begin{enumerate}
		\item A \(\bigl(\alpha_B,a_B,0\bigr)\)-block-encoding of \(B\), denoted \(U_B\), and
		\item A \(\bigl(1,\,a_B+1,\,0\bigr)\)-block-encoding of \(\dyad{0}{1}\otimes B\), denoted \(U_B^{(01)}\).
	\end{enumerate}
	The combined circuit (depicted in Fig.~\ref{circ:BE_H}) is 
	\begin{equation}
		U_{\mathcal{H}}=\bigl(H\otimes I_{n+a_B+2}\bigr)\Bigl[\dyad{0}\otimes U_B^{(01)\dagger} +\dyad{1}\otimes U_B^{(01)}\Bigr]
		\bigl(H\otimes I_{n+{a_B}+2}\bigr), \label{eq:BE_H}
	\end{equation}
	yielding a \(\bigl(\alpha_{\mathcal{H}},\,a_{\mathcal{H}},\,0\bigr)\)-block-encoding of \(\mathcal{H}\), where \(\alpha_{\mathcal{H}} = 2, \alpha_B\) and \(a_{\mathcal{H}} = a_B + 2\). (Here, the extra factor of 2 arises from the sum of \(\dyad{0}{1}\otimes B\) and its Hermitian conjugate, and the gate \(R_y(2\pi)\) in the circuit effectively introduces a global \(-1\) factor.)
	
	Since the complexity of block-encoding \(\mathcal{H}\) directly inherits that of \(B\), we have
	\begin{equation}
		G_{\mathcal{H}}(N) 
		=\begin{cases}
			\mathcal{O}\bigl(\log_2^2 N\bigr), 
			& \text{(uniform system)}\\[6pt]
			\mathcal{O}\bigl(2N\log_2 N +\log_2^2 N \bigr).
			& \text{(generalized system)} 
		\end{cases} 
		\numberthis \label{eq:GH}
	\end{equation}
	Thus, once \(B\) is encoded, obtaining a block-encoding of \({\mathcal{H}}\) requires only a modest overhead from LCU composition. We now proceed to show how \(e^{-\,i\mathcal{H}t}\) can be approximated via QSVT using this block-encoding.
	
	\subsection{Block Encoding Circuit of \boldmath\(\exp(-i\mathcal{H}t)\)}
	\label{subsec:BE_eiht}
	
	Having obtained a block-encoding circuit of $\mathcal{H}$, we now realize the time-evolution operator \(e^{-\,i\mathcal{H}\,t}\) via QSVT. While QSVT can approximate a polynomial (or related trigonometric expansions) of a block-encoded operator, certain constraints arise  (cf.  the QSVT constraints in Section \ref{subsec:QSVT}.):
	\begin{enumerate}[label=(\roman*)]
		\item the target polynomial must be supported on a positive-definite domain,
		\item it must have a definite parity (odd or even), and 
		\item it must be magnitude-bounded by 1.
	\end{enumerate}
	
	\paragraph*{1. Shifting $\mathcal{H}$ to a positive operator.}
	Because the eigenvalues of the block-encoded $H$ may lie in $[-1,1]$, we cannot directly apply QSVT to $\mathcal{H}$ itself. Instead, we introduce
	\[
	\hat{\mathcal{H}}
	\;\coloneqq\;
	\tfrac{1}{2}\,\bigl(\mathcal{H}/\alpha_{\mathcal{H}} + I\bigr),
	\]
	so that \(\hat{\mathcal{H}}\) is strictly positive and satisfies \(\|\hat{\mathcal{H}}\|\ge1\). In practice, we implement a \(\bigl(1,a_{\mathcal{H}}+1,0\bigr)\)-block-encoding of \(\hat{\mathcal{H}}\) by linearly combining $U_\mathcal{H}$ (the block-encoding of $\mathcal{H}$) with the identity operator.
	
	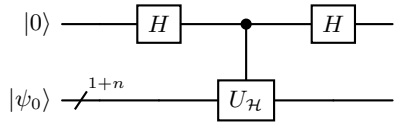
\begin{figure}[ht!]
		\centering
		\begin{quantikz}
			\lstick{$\ket{0}$}	&&\gate{H}	&\ctrl{1}	&\gate{H}	&\\
			\lstick{$\ket{\psi_0}$}	&\qwbundle{1 + n}&&\gate{U_{\mathcal{H}}}& &
		\end{quantikz}
		\caption{{\bf A $(1,\,a_{\mathcal{H}}+1,\,0)$-block-encoding of 
				\(\hat{\mathcal{H}} \equiv \tfrac{1}{2}\bigl(\mathcal{H}/\alpha_{\mathcal{H}} + I\bigr)\).} 
			We use LCU to combine \(\mathcal{H}/\alpha_{\mathcal{H}}\) and \(I\).}
		\label{circ:BE_H_tilde}
	\end{figure}
	
	\paragraph*{2. Enforcing polynomial parity and unit magnitude.}
	We wish to approximate 
	\[
	e^{-\,i\mathcal{H}t}
	\;=\;
	\cos\bigl(\mathcal{H}\,t\bigr)
	-
	i\;\sin\bigl(\mathcal{H}\,t\bigr).
	\]
	Since $\cos(\cdot)$ is even and $\sin(\cdot)$ is odd, both fit QSVT’s parity requirement. However, because QSVT also necessitates the polynomial to have amplitude at most 1 on the domain, we first rescale these trigonometric functions to 
	\[
	\tfrac{1}{2}\,\cos(\hat{\mathcal{H}}\,\tau)
	\quad\text{and}\quad
	-\tfrac{i}{2}\,\sin(\hat{\mathcal{H}}\,\tau),
	\]
	where we defined the {\it effective time} \(\tau \coloneqq 2\,\alpha_{\mathcal{H}}\,t\). This ensures that each subroutine is bounded by 1, allowing us to build separate block-encodings for \(\cos(\hat{\mathcal{H}}\tau)\) and \(\sin(\hat{\mathcal{H}}\tau)\). 
	In practice, we first compute block encodings of \(\cos(\hat{\mathcal{H}}\,\tau)\) and \(-i\,\sin(\hat{\mathcal{H}}\,\tau)\) using phase angles generated by \texttt{QSPPACK}.\texttt{QSPPACK} provides a set of phase angles to approximate trigonometric polynomials (or other target functions) via QSVT up to error \(\epsilon\). This yields polynomial approximations \(P_{\texttt{QSPPACK}}^{\cos}\) and \(P_{\texttt{QSPPACK}}^{\sin}\) that each lie within \(\epsilon\) of the target functions $\cos(\cdot)$ and $\sin(\cdot)$, respectively.
	
	Next, we combine these two encodings with an LCU step to form 
	\[
	\tfrac12 \times \tfrac12\Bigl[\cos(\hat{\mathcal{H}}\,\tau)\;-\;i\,\sin(\hat{\mathcal{H}}\,\tau)\Bigr]
	\;=\;
	\tfrac14\,e^{-\,i\,\hat{\mathcal{H}}\,\tau}.
	\]
	Note that in the \(-\,i\,\sin\) block encoding, the factor of \(-i\) can be implemented by a simple sequence of single-qubit gates (\(X\) then \(Y\)) on an ancilla qubit. Also, the extra  factor of $1/2$ comes from the state-preparation in the LCU step.
	
	Finally, we apply an \(R_z(-\tau)\) gate to remove the global phase factor \(e^{-\,i\tau/2}\). Substituting back \(\hat{\mathcal{H}} = \bigl(\mathcal{H}/\alpha_{\mathcal{H}} + I\bigr)/2\) and \(\tau=2\,\alpha_{\mathcal{H}}\,t\) yields our final \(\bigl(\alpha_{\text{HS}},\,a_{\mathcal{H}}+3,\,\epsilon\bigr)\)-block-encoding of \(e^{-\,i\,\mathcal{H}\,t}\) where $\alpha_{\text{HS}}$ is 4 in our work. Concretely,
	\[
	\tfrac14\,e^{-\,i\,\hat{\mathcal{H}}\,\tau}
	\;\;\xrightarrow{R_z(-\tau)}\;\;
	\frac{\tfrac14\,e^{-\,i\,\hat{\mathcal{H}}\,\tau}}{\,e^{-\,i\tau/2}\,}
	\;=\;
	\tfrac14\,e^{-\,i\,{\mathcal{H}}\,t}
	\]
	up to a known sub-normalization.  Three additional ancilla qubits arise from: 
	1) the LCU step needed to transform \(\hat{\mathcal{H}}\), 
	2) the QSVT-based construction of the \(\cos\) and \(\sin\) block encodings, 
	and 3) the final LCU procedure that assembles the time-evolution operator. Figure~\ref{circ:BE_HS_H} illustrates the core circuit.
	
	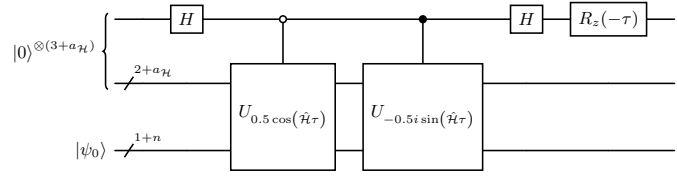
\begin{figure}[ht!]
		\centering
		\begin{adjustbox}{width=0.5\textwidth}
			\begin{quantikz}
				\lstick[2]{$\ket{0}^{\otimes (3 + a_{\mathcal{H}})}$} &					&\gate{H}	&\ctrl[open]{1}	&\ctrl{1}	&\gate{H} &\gate{R_z(-\tau)}&\\
				&\qwbundle{2 + a_{\mathcal{H}}}&	&\gate[2]{U_{0.5\cos(\hat{\mathcal{H}}\tau)}}	&\gate[2]{U_{-0.5i \sin(\hat{\mathcal{H}}\tau)}} & && \\
				\lstick{$\ket{\psi_0}$}					&\qwbundle{1 + n}		&&&&&&
			\end{quantikz}
		\end{adjustbox}
		\caption{{\bf A \(\bigl(\alpha_{\text{HS}},a_{\mathcal{H}}+3,\epsilon\bigr)\)-block-encoding of \(e^{-\,i\mathcal{H}\,t}\) where $\alpha_{\text{HS}} = 4$.} We construct approximations for \(\cos(\hat{\mathcal{H}}\,\tau)\) and \(-\,i\,\sin(\hat{\mathcal{H}}\,\tau)\) (with phase angles from \texttt{QSPPACK}), recombine them via LCU into \(\tfrac{1}{2}e^{-\,i\hat{\mathcal{H}}\,\tau}\), and then apply an \(R_z(-\tau)\) to remove an extra phase \(e^{-\,i\tau/2}\), with \(\tau=2\,\alpha_{\mathcal{H}}\,t\).}
		\label{circ:BE_HS_H}
	\end{figure}

	\paragraph*{3. Circuit complexity.}
	We now discuss how the overall circuit cost scales with the simulation time \(t\) and the desired accuracy \(\epsilon\). In particular, we must approximate \(\cos(\hat{\mathcal{H}}\,\tau)\) and \(\sin(\hat{\mathcal{H}}\,\tau)\) by suitable polynomials via QSVT, where \(\tau = 2\,\alpha_{\mathcal{H}}\,t\). Following the analysis in Ref.~\cite{Dong_2021} (which adopts Jacobi--Anger expansions of sine and cosine functions), the required polynomial degree \(k(t,\epsilon)\) satisfies 
	\begin{equation}
		k(t,\epsilon)
		\;\approx\;
		\Bigl\lceil\,1.4\,\tau\;+\;\log\bigl(1/\epsilon\bigr)\Bigr\rceil,
	\end{equation}
	reflecting that longer time \(t\) and smaller error \(\epsilon\) both demand a higher-degree expansion of \(\sin\) and \(\cos\) in order to maintain the target approximation fidelity over the relevant domain. Note that, for $\epsilon \in (0,0.5)$ and $t \in \mathbb{R}$, a more precise (though less operational) approximation is given by $k(t,\epsilon) = \Theta \qty\Bigg(\alpha_{\mathcal{H}}\abs{t} + \frac{\log(1/\epsilon)}{\log (e + \frac{\log(1/\epsilon)}{\alpha_{\mathcal{H}}\abs{t}})})$.     
	
	In practice, we obtain the phase angles for these polynomial approximations using the \texttt{QSP\_Solver} routine from \texttt{QSPPACK} (in MATLAB), which may shift the sine or cosine degrees by 1 to accommodate parity:
	\begin{equation*} \label{eq:QET_MATLAB_dSin}
		d_{\sin} = 
		\begin{cases}
			k 		& \text{if $k$ is even}, \\
			k-1 	& \text{if $k$ is odd},
		\end{cases}
	\end{equation*}
	\begin{equation*} \label{eq:QET_MATLAB_dCos}
		d_{\cos} = 
		\begin{cases}
			k-1 		& \text{if $k$ is even}, \\
			k 	& \text{if $k$ is odd},
		\end{cases}
	\end{equation*}
	Hence, implementing both the sine and cosine expansions demands \(d_{\sin}+d_{\cos}=2\,k-1\) applications of the controlled block-encoding \(\mathrm{CU}_{\hat{\mathcal{H}}}\), alongside only \(\mathcal{O}(1)\) gates for the \((a_{\mathcal{H}}+3)\)-qubit projector-controlled phase operations. Notably, \(\mathrm{CU}_{\hat{\mathcal{H}}}\) shares the same gate complexity \(G_{\mathcal{H}}\) as the underlying (uncontrolled) block-encoding \(U_{\hat{\mathcal{H}}}\). Thus, constructing the \(\bigl(4,a_{\mathcal{H}}+3,0\bigr)\)-block-encoding of \(e^{-i\mathcal{H}t}\) to precision \(\epsilon\) has the gate  complexity.
	\begin{equation}
		\label{eq:gateCom_HS_explicit}
		G_{e^{-i\mathcal{H}\,t}}(t,\epsilon)
		\;=\;
		\Bigl(2\,\bigl\lceil 2.8\,\alpha_{\mathcal{H}}\,t \;+\;\log\bigl(1/\epsilon\bigr)\bigr\rceil - 1\Bigr)\;G_{\mathcal{H}}.
	\end{equation}
	Recall from \eqref{eq:GH} that \(G_{\mathcal{H}}=\mathcal{O}(\log_2^2 N)\) in the uniform case, or \(G_{\mathcal{H}}=\mathcal{O}(2\,N\,\log_2 N+\log_2^2 N)\) for generalized masses and spring constants. Thus, the final gate complexity for an \(\epsilon\)-close time-evolution circuit at simulation time~\(t\) scales as
	\[
	\mathcal{O}\Bigl[\bigl(\alpha_{\mathcal{H}}\,t+\log(1/\epsilon)\bigr)\;G_{\mathcal{H}}\Bigr].
	\]
	Table~\ref{table:gateCom_timeEvol} summarizes the key gate counts in this final QSVT-based construction.
	
	\begin{table}[!ht]
		\centering
		\caption{{\bf Gate counts in constructing the \(\bigl(4,a_H+3,\epsilon\bigr)\)-block-encoding circuit of \(e^{-\,i\mathcal{H}\,t}\).} 
			Here, \(k\approx 1.4\,\tau+\log(1/\epsilon)\) denotes the polynomial degree.}
		\label{table:gateCom_timeEvol}
		\setlength{\tabcolsep}{3pt}
		\begin{tabular}{|c|c|}
			\hline
			\textbf{Gate type} & \textbf{Number of calls}\\\hline
			\(R_z\) & 1\\\hline
			\(H\) & \(2+2(2k-1)\)\\\hline
			$\mathrm{C}R_z$ & 1\\\hline
			$\mathrm{C}Z$ & 1\\\hline
			$\mathrm{CU}_{\mathcal{H}}$ & $2k-1$\\\hline
			$\mathrm{C}\Pi(a_{\mathcal{H}}+3)$ & $2k-1$\\\hline
		\end{tabular}
	\end{table}

	\subsection{Full Quantum Algorithm for Hamiltonian Simulation}
	\label{subsec:fullHS}
	
	\begin{figure*}
		\centering
		\begin{adjustbox}{width=0.95\textwidth}
			\begin{quantikz}[slice style = blue]
				\lstick{$\ket{0}^{\otimes a}$} &\qwbundle{a} & \slice{$\ket{0}^{\otimes a}\ket{\psi_0}$} &\gate[2]{U_{e^{-i{\mathcal{H}}t}}} 	& \gate{R_{\psi_g} = R_0(a)}\gategroup[2, steps=6, style={dashed, rounded corners, fill = blue!20, inner xsep = 2pt}, background, label style={label position = above, anchor = north, yshift=0.4cm}]{\textcolor{blue}{$W$}} 	&\gate[2]{U_{e^{-i{\mathcal{H}}t}}^\dagger} 	& \gategroup[2, steps=3, style={dashed, rounded corners, fill = red!20, inner sep = 1pt}, background, label style={label position = above, anchor = north, yshift=0.42cm}]{\textcolor{red}{$R_{\psi_0}$}}						& \gate[2]{R_{0}(a+n+1)} & 			& \gate[2]{U_{e^{-i{\mathcal{H}}t}}} & \meter{} \rstick[brackets=none]{$\ket{0}^{\otimes a}$} \\
				\lstick{$\ket{0}^{\otimes n+1}$} &\qwbundle{n+1} & \gate{S} 	& 				& 								& 						& \gate{S^\dagger} 	& 			 		&\gate{S} & 				& \rstick[brackets=none]{$\frac{e^{-i{\mathcal{H}}t} \ket{\psi_0}}{\norm{e^{-i{\mathcal{H}}t} \ket{\psi_0}}_2}$}
			\end{quantikz}
		\end{adjustbox}
		\caption{{\bf A schematic of the full circuit implementing of the time-evolution operator \(e^{-i\mathcal{H}t}\) with \emph{robust oblivious amplitude amplification} (ROAA).} The operator \(U_{e^{-i\mathcal{H}t}}\) is the \(\bigl(\alpha_{\text{HS}},a_{\mathcal{H}}+3,\epsilon\bigr)\)-block-encoding of \(e^{-i\mathcal{H}t}\), with $\alpha_{\text{HS}} = 4$, described in Section \ref{subsec:BE_eiht}. A state-preparation gate \(S\) initializes \(\ket{\psi_0}\) from \(\ket{0}^{n+1}\). The ROAA machinery (center box), iterated for $Q_W$ times, boosts the amplitude of the “good” outcome, ensuring that measuring the ancillas in \(\ket{0}^{a}\) returns the properly time-evolved state \(\tfrac{e^{-i\mathcal{H}t}\ket{\psi_0}}{\|\,e^{-i\mathcal{H}t}\ket{\psi_0}\|}\) with high probability.}
		\label{circ:BE_fullHS_AA}
	\end{figure*}
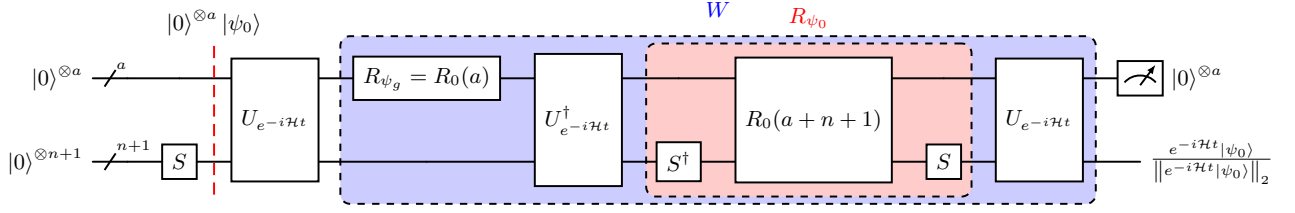
	Having constructed a block-encoding of \(e^{-\,i\mathcal{H}\,t}\), 
	our next step is to amplify the probability of projecting the ancilla qubits onto 
	\(\ket{0^a}\) (so that the signal register is left in the desired time-evolved state). 
	Whereas standard amplitude amplification suffices for many applications \cite{grover1996fast}, 
	the \emph{robust oblivious amplitude amplification} (ROAA) protocol~\cite{Berry_Childs_Cleve_Kothari_Somma_2014} 
	often proves more efficient for block-encoded matrices~(see Section~\ref{subsec:ROAA}).

	\paragraph*{1. ROAA procedure.}
	Let \(\ket{\psi_0} \coloneq \ket{\psi(0)}\) of the signal register be the initial state encoding the initial condition of our spring-mass system defined according to \eqref{eq:CO_state}, and suppose we have a 
	\(\bigl(\alpha_{\mathrm{HS}},\,a,\,\epsilon\bigr)\)-block-encoding of \(e^{-\,i\mathcal{H}\,t}\), 
	which we denote \(U_{e^{-\,i\mathcal{H}\,t}}\). Acting on \(\ket{\psi_0}\) with this block-encoding 
	(and measuring the ancillas) ideally yields \(e^{-\,i\mathcal{H}\,t}\ket{\psi_0}\). However, 
	the measurement success probability may be small. To boost it, we apply ROAA 
	with cost proportional to
	\begin{equation}
		Q_{\mathrm{W}}\bigl(\alpha_{\mathrm{HS}}\bigr)
		\;=\;
		\mathcal{O}\qty({\frac{\alpha_{\mathrm{HS}}}{\norm{\ket{\psi_0}}_2}}),
		\label{eq:timeCom_W}
	\end{equation}
	where $W$ denotes the Grover iteration operator defined in Section \ref{subsec:ROAA} and $\alpha_{\mathrm{HS}} = 4$ in our work. This quantity captures how many amplitude-amplification queries (iterations) are required to guarantee a high-probability measurement of the desired time-evolved state \(e^{-\,i\mathcal{H}\,t}\ket{\psi_0}\) as shown in Ref. \cite{Berry_Childs_Cleve_Kothari_Somma_2014}. Further, each $W$ calls the block-encoding \(U_{e^{-\,i\mathcal{H}\,t}}\) twice, 
	so the total query count, including the leftmost one, is
	\begin{equation} 
		Q_{\mathrm{AA}}\bigl(\alpha_{\mathrm{HS}}\bigr)
		\;=\;
		2\;Q_{\mathrm{W}}\bigl(\alpha_{\mathrm{HS}}\bigr)\;+\;1.    
		\label{eq:timeCom_AA}
	\end{equation}
	
	\paragraph*{2. Putting it all together.}
	Combining the above with our earlier estimate for constructing \(U_{e^{-\,i\mathcal{H}\,t}}\), 
	the overall cost \(G_{\mathrm{HS+AA}}\bigl(\alpha_{\mathrm{HS}},\,t,\,\epsilon\bigr)\) 
	at a single simulation time \(t\) with the error $\epsilon$ scales as
	\begin{equation}
		G_{\mathrm{HS+AA}}\bigl(\alpha_{\mathrm{HS}},\,t,\,\epsilon\bigr)
		\;=\;
		Q_{\mathrm{AA}}\bigl(\alpha_{\mathrm{HS}}\bigr)\;
		G_{e^{-\,i\mathcal{H}\,t}}\bigl(t,\,\epsilon\bigr),
		\label{eq:timeCom_HS_AA_completed}
	\end{equation}
	where \(G_{e^{-\,i\mathcal{H}\,t}}\) is the gate complexity for building 
	\(\bigl(\alpha_{\mathrm{HS}},\,a_{\mathcal{H}}+3,\,\epsilon\bigr)\)-block-encodings of \(e^{-\,i\mathcal{H}\,t}\). 
	Figure~\ref{circ:BE_fullHS_AA} summarizes the final quantum circuit, with $a=a_{\mathcal{H}}+3$ 
	as the total number of ancillas. 
	
	We now state our main result for simulating the coupled‐oscillator Hamiltonian \(\mathcal{H}\) up to time~\(t\), and obtaining the correct time‐evolved state with high probability.
	
	\begin{theorem}\label{thm:tsim_t}
		Let \(\alpha_{\mathcal{H}}\) be the sub‐normalization constant of the block‐encoding for \(\mathcal{H}\), and let \(\alpha_{\mathrm{HS}}\) be the sub‐normalization constant of the block‐encoding for \(e^{-\,i\mathcal{H}\,t}\). Suppose the initial state is \(\ket{\psi_0}\). Then, we can implement an \(\epsilon\)-close Hamiltonian simulation at simulation time~\(t\) with overall gate complexity
		\begin{equation}
			G_{\mathrm{HS+AA}}(\alpha_{\mathcal{H}},\alpha_{\mathrm{HS}},\epsilon; t) =  \qty({2 \Bigg\lceil \qty({\frac{\abs{\alpha_{\mathrm{HS}}}}{\norm{\psi_0}_2}}) \Bigg\rceil} + 1)
			\times \qty({2 \bigg\lceil 2.8\alpha_{\mathcal{H}}t + \log(\frac{1}{\epsilon}) \bigg\rceil - 1}) \times G_{\mathcal{H}} \label{eq:timeCom_HS_AA}
		\end{equation}
		where \(G_{\mathcal{H}}\) is the gate complexity of block‐encoding \(\mathcal{H}\) (see ~\eqref{eq:GH}).
	\end{theorem}
	
	In ~\eqref{eq:timeCom_HS_AA}, the first factor in parentheses arises from ROAA, ensuring a high‐probability measurement of the time‐evolved state (cf.\ Eqs.~\eqref{eq:timeCom_W}–\eqref{eq:timeCom_AA}). The second factor follows the Hamiltonian‐simulation cost, given by our earlier analysis (cf. ~\eqref{eq:gateCom_HS_explicit}). 
	See Fig.~\ref{circ:BE_fullHS_AA} for an example circuit implementing this procedure.

	\paragraph*{3. Extending to multiple time steps.}
	In practice, one may wish to extract the time-evolved state at multiple time points \(t_i,\,t_i+\Delta t,\,\dots,\;t_f = t_i + M \Delta t\). At each step \(t_i + j\,\Delta t\), we apply the block-encoded time-evolution operator \(e^{-i\mathcal{H}\,(t_i + j\,\Delta t)}\) and then perform ROAA to amplify the success probability of measuring the desired state. This procedure, repeated for \(j=0,1,\dots,M\), is precisely what we implement in the numerical simulations in Section~\ref{sec:results_exp}. 
	
	The overall gate complexity simply accumulates the single-time-evolution cost \eqref{eq:timeCom_HS_AA} over all \(M+1\) increments, as formalized below
	
	\begin{corollary}\label{thm:tsim_full}
		Let \(\alpha_{\mathcal{H}}\) be the sub-normalization constant of the Hamiltonian \(\mathcal{H}\)'s block-encoding, and let \(\alpha_{\mathrm{HS}}\) be that of the time-evolution operator \(e^{-i\mathcal{H}\,t}\). Then, to simulate from an initial time \(t_i\) to a final time \(t_f\) in uniform steps of size \(\Delta t\), the total cost of an \(\epsilon\)-close Hamiltonian simulation is
		\begin{equation}
			G_{\mathrm{sim}}(\alpha_{\mathcal{H}},\alpha_{\mathrm{HS}},\epsilon; t_i,\Delta t, M) = 
			\sum_{j=0}^{M} \Bigg[G_{\mathrm{HS+AA}}(\alpha_{\mathcal{H}},\alpha_{\mathrm{HS}},\epsilon; t_i + j\Delta t)\Bigg] \numberthis \label{eq:timeCom_HS_AA_FULL}
		\end{equation}
		where \(M \coloneqq (t_f - t_i)/\Delta t\), and \(G_{\mathrm{HS+AA}}\) follows from Theorem~\ref{thm:tsim_t}. 
	\end{corollary}
	
	Hence, by iterating block-encoded time evolution and amplitude amplification at each time slice, one advances the system from \(t_i\) to \(t_f\) in discrete increments.  The final state at each step is amplified to high probability on the ancillas, yielding the correct time-evolved state of the signal qubits. 
	
	\section{Simulation Results}
	\label{sec:results_exp}
	
	\begin{figure*}[ht!]
		\centering
		\includegraphics[width = 0.9\textwidth]{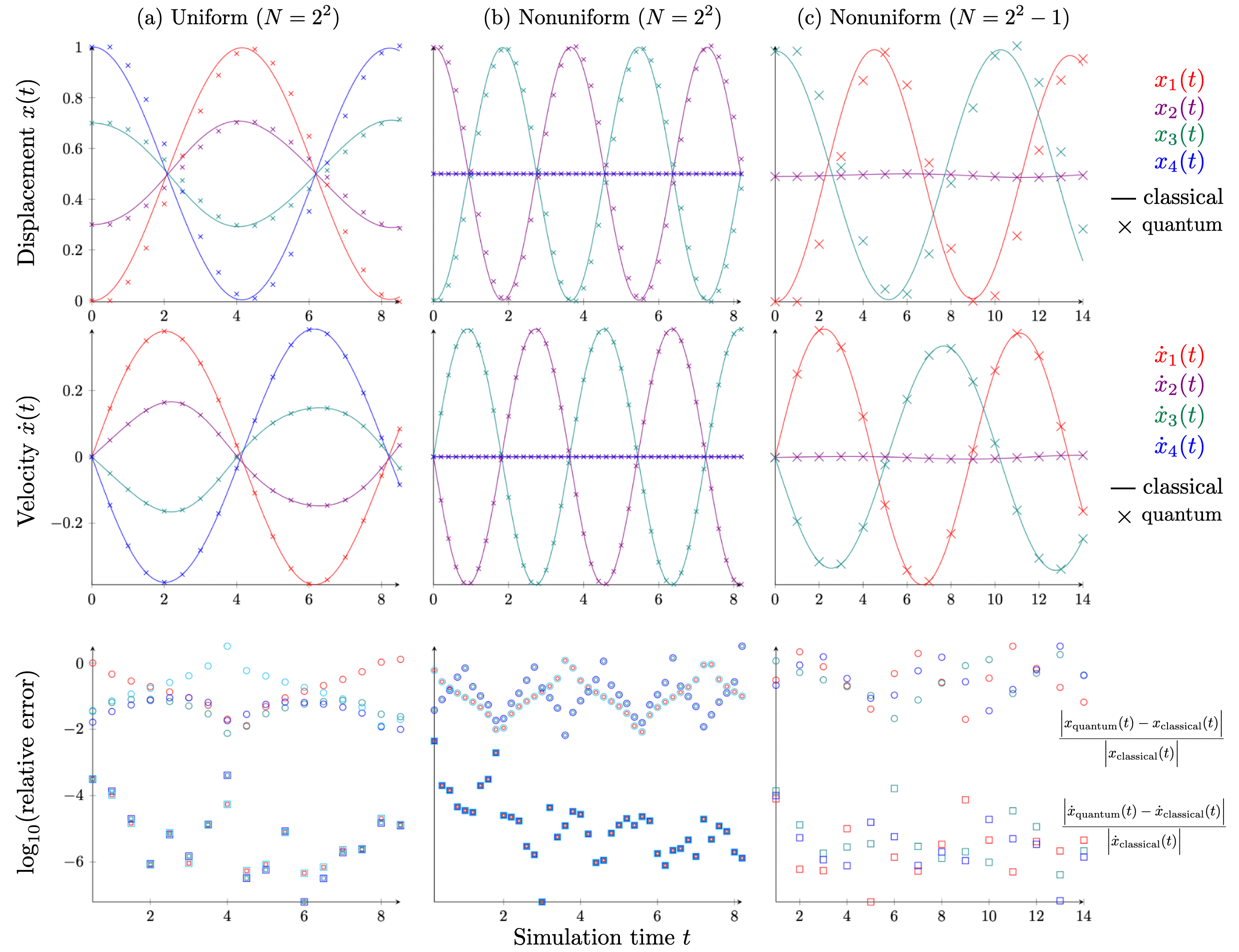}
		\caption{{\bf Quantum–classical comparisons of displacement (top row), velocity (middle row), and logarithmic relative error (bottom row) for three example systems:} 
			(a)~A \emph{uniform} chain of four oscillators (all masses and springs set to~1) with open boundaries and sampled from \(t_i=0\) to \(t_f=8.5\) at \(\Delta t=0.5\). 
			(b)~A \emph{nonuniform} four-oscillator chain in which the first and last masses (\(m_1=m_4=99999\)) effectively act as walls, leaving two light interior masses \((m_2,m_3=1)\) vibrating. Here, the time evolution runs to \(t=8.2\) in increments of \(\Delta t=0.2\). 
			(c)~A \emph{nonuniform} three-oscillator system \((N=3\neq2^n)\) with one large central mass \((m_2=100)\) and two lighter end masses \((m_1=1,m_3=2)\), simulated until \(t=14\) in steps of \(\Delta t=1.0\). 
			In the top and middle row, colored lines denote classical Runge–Kutta solutions, while crosses represent the quantum results obtained from our QSVT-based time-evolution circuit at an approximation error \(\epsilon=0.01\). In the bottom row, we plot the relative errors for both displacement (\(\circ\)) and velocity (\(\square\)) in $\log_{10}$ scale. The velocity error consistently remains near \(\mathcal{O}(10^{-5})\), whereas displacements show slightly larger discrepancies (\(\mathcal{O}(10^{-1})\)) due to the finite-difference post-processing. Reducing \(\Delta t\) can further improve the displacement accuracy (see Section~\ref{sec:results_exp} for the discussion).}
		\label{fig:result}
	\end{figure*}

	In this section, we present numerical demonstrations of our proposed quantum algorithm using the Python library \texttt{Qiskit}~\cite{Qiskit}. We prepare the initial state vectors via the built-in state-preparation gates in \texttt{Qiskit} and retrieve the final statevector using the \texttt{Statevector} class. Note that, in practice, one must perform quantum state tomography to obtain a complete statevector on hardware, which typically costs some constants to the power of $N$ (see Ref. \cite{Stricker_2022}.) The examples below assume an idealized simulator for clarity.
	
	Because the matrix \(B\) (cf.\ Section~\ref{sec:CCO}) is often singular, we can reliably extract only the \emph{velocity} degrees of freedom \(\dot{\vec{x}}(t)\) directly from the measured statevector. To recover the displacement \(\vec{x}(t)\), we employ a finite-difference method, yielding truncation errors of order \(\mathcal{O}(\epsilon\,\Delta t)\) in displacement (where \(\epsilon\) is the QSVT approximation error and \(\Delta t\) is the time step). Alternatively, one could seek to implement an \emph{end-to-end} circuit that uses a Moore--Penrose pseudoinverse of \(B\) to recover directly the displacement vector, bypassing finite-differences.
	
	\subsubsection{Unit masses and springs (\(N=4\))}
	\label{sss:SR_unit}
	We first simulate a chain of four oscillators (\(j=1,2,3,4\)) with unit masses and springs and open-boundary conditions (i.e.\ no coupling between the first and last oscillator). Initially, the displacements are
	\[
	\vec{x}(0) 
	\;=\;
	\begin{pmatrix}
		0\\
		0.3\\
		0.7\\
		1.0
	\end{pmatrix},
	\quad
	\dot{\vec{x}}(0)
	\;=\;
	\vec{0},
	\]
	with the physical time ranging from \(t_i = 0\) to \(t_f = 8.5\), sampled in increments of \(\Delta t=0.5\). The QSVT truncation error is set to \(\epsilon=0.01\), while the classical loop terminates if the residual drops below \(t_{\mathrm{tol}}=10^{-5}\).
	
	Figure~\ref{fig:result} compares the displacements and velocities from the quantum simulation versus a classical Runge--Kutta solver. We quantify agreement by a logarithmic relative error. The relative velocity error hovers around \(10^{-5}\), an extremely small mismatch, whereas the displacement error remains closer to \(10^{-1}\). This discrepancy follows naturally from the finite-difference approach used to recover \(\vec{x}(t)\). Reducing \(\Delta t\) lowers that error.
	
	\subsubsection{Nonuniform system (\(N=2^2=4\)) with ``walls''}
	\label{sss:SR_gen_2n}
	Next, we consider again an open chain of four oscillators, but set the end masses to \(m_1=m_4=99999\) and the inner masses \(m_2=m_3=1\). Effectively, the outer oscillators act like very large walls, yielding negligible motion. We keep the spring constants uniform (all~1) and initialize the inner oscillators with antisymmetric displacements: \(x_2(0)=+1\) and \(x_3(0)=-1\). Simulations proceed up to \(t=8.2\) in steps of \(\Delta t=0.2\), with a QSVT error \(\epsilon=0.01\). Figure~\ref{fig:result} illustrates the outcome. As before, the relative error in velocity lies around \(10^{-5}\), while displacement errors are around \(10^{-1}\). 
	
	\subsubsection{Nonuniform system (\(N=3 \neq 2^n\))}
	\label{sss:SR_gen_NOT2n}
	Finally, we consider a three-oscillator system \(j=1,2,3\) with masses \(m_1=1\), \(m_2=100\), \(m_3=2\). We choose \(k_{(1,2)}=0.5\), \(k_{(2,3)}=0.75\), and \(k_{(3,1)}=0\). Setting \(k_{(3,1)}=0\) specifies an open boundary between oscillator~3 and oscillator~1. The large middle mass serves as a wall, while the other two masses vibrate at the ends. We initialize 
	\[
	\vec{x}(0)
	\;=\;
	\begin{pmatrix}
		-1\\[3pt]
		0\\[3pt]
		+1
	\end{pmatrix},
	\quad
	\dot{\vec{x}}(0)
	\;=\;\vec{0},
	\]
	then evolve out to \(t=8\) in increments of \(\Delta t=0.2\) with a QSVT precision of \(\epsilon=10^{-2}\). As shown in Fig.~\ref{fig:result}, the displacements again exhibit relative errors of roughly \(10^{-1}\), while the velocities are accurate to about \(10^{-5}\). 
	
	These small-scale tests confirm that the quantum algorithm faithfully reproduces classical trajectories, with displacement errors mainly arising from the finite-difference post-processing. Decreasing the time-step \(\Delta t\) decreases these errors by refining the finite-difference estimate of the oscillator displacements. Furthermore, if a block encoding of \((B^\dagger)^{-1}\) is implementable, one can apply it directly to the initial state \(\ket{\psi_0}\), circumventing the need for any subsequent post-processing of \(\vec{x}(t)\). We leave this explicit block encoding of $B^\dagger$ for future work.
	
	
	\section{Conclusion and Outlook}
	\label{sec:conc}
	
	We have presented a concrete quantum-circuit construction for simulating the dynamics of one-dimensional classical coupled-oscillator networks. Our methodology combines block encoding, quantum singular-value transformation (QSVT), and robust oblivious amplitude amplification (ROAA) to implement the time-evolution operator \(e^{-i\mathcal{H}t}\) of the system’s Hamiltonian. Our approach applies to any number \(N\) of oscillators. Although our scheme initially requires masses and spring constants in \([1,\infty)\) and \((0,1]\), respectively, out-of-range values can be accommodated via rescaling at the cost of increased sub-normalization in the block-encoding circuits.
	
	We constructed a block encoding of \(\mathcal{H}\) for general masses and spring constants in \(\mathcal{O}(2\,N\log_2 N + \log_2^2 N)\) gates, which then enables a QSVT-based approximation of the time-evolution operator.  ROAA then amplifies the probability of successfully measuring the correct time-evolved state. These steps together yield an overall gate complexity of \(\mathcal{O}\bigl(\log(1/\epsilon)\,G_{\mathcal{H}}\,(t_f - t_i)/\Delta t\bigr)\) to achieve \(\epsilon\)-accuracy over the time interval \([t_i, t_f]\) with timestep \(\Delta t\), where \(G_{\mathcal{H}}\) is the cost of block encoding \(\mathcal{H}\). Numerical results confirm that the quantum-simulated velocities and displacements agree well with classical integrators, showing typical errors down to \(10^{-5}\) in velocity and \(10^{-1}\) in displacement, due largely to the finite-difference reconstruction of the displacement from the velocity. 
	
	Critically, the gate complexity depends strongly on the cost of block encoding \(\mathcal{H}\). While it scales favorably as \(\mathcal{O}\bigl(\log_2^2 N\bigr)\) for uniform systems, it is \(\mathcal{O}\bigl(2\,N \log_2 N\bigr)\) for general masses and springs systems. Although the qubit count grows only logarithmically with \(N\), scaling to the thermodynamic limit remains gate-intensive for the general system. Future efforts could explore more efficient block encoding of \(\mathcal{H}\) that push complexity toward sublinear or logarithmic scaling, enabling larger-scale quantum simulations.
	
	Although our numerical experiments use an ideal, noise-free simulator (Qiskit~\cite{Qiskit}), hardware realization will require substantial error mitigation or fault tolerance, given the circuit depths associated with QSVT and ROAA. Nonetheless, for sufficiently low-noise devices and \emph{uniform} oscillator parameters, our approach suggests exponential speedups over classical time-stepping. Further refinements could incorporate a QSVT-based pseudoinverse of the matrix \(B^\dagger\) \cite{Gily_n_2019} directly into the quantum circuit, reducing the cost of classical post-processing, as the final state would already encode the required displacements.
	
	Looking forward, a natural direction is to generalize our approach from one-dimensional chains to higher-dimensional oscillator networks. Crucially, the network connectivity governs the incidence matrix \(\Phi_F\), whose block encoding underlies the Hamiltonian simulation. In two- or three-dimensional grids with uniform spring constants, one may exploit the well-known \emph{Kronecker-sum} formulation of discrete Laplacians to build the higher-dimensional analog of the matrix~\(F\). For instance, a 2D grid of size \(M\times N\) admits
	\[
	F_{\mathrm{2D}}
	=
	F_{\mathrm{1D},M}\,\otimes\,I_{N}
	+
	I_{M}\,\otimes\,F_{\mathrm{1D},N},
	\]
	while a 3D grid of size \(M\times N\times P\) can be written as
	\[
	F_{\mathrm{3D}}
	=
	F_{\mathrm{1D},M}\,\otimes\,I_{N}\,\otimes\,I_{P}
	+
	I_{M}\,\otimes\, F_{\mathrm{1D},N}\,\otimes\,I_{P}
	+
	I_{M}\,\otimes\,I_{N}\,\otimes\,F_{\mathrm{1D},P},
	\]
	where \(F_{\mathrm{1D},K}\) is the standard one-dimensional discrete Laplacian on \(K\) nodes. From there, one factorizes \(F_{\mathrm{2D}}\) or \(F_{\mathrm{3D}}\) into an incidence matrix \(\Phi_F\) and diagonal weight factors, analogous to the one-dimensional case. Constructing a block encoding of \(\Phi_F\) thus extends our scheme to multidimensional settings. Beyond spatially extended lattices, one could further consider more general network topologies where a similar Laplacian-based approach may provide an advantage when simulated using quantum circuits. Ultimately, such multidimensional simulations could leverage quantum hardware to tackle computationally heavy tasks (e.g.\ large-scale partial differential equations ~\cite{Jin_liu_comphys_2024, Jin_liu_QST_2024}), highlighting the promise of circuit-based quantum simulation of classical dynamics. We leave a rigorous treatment of these higher-dimensional and general-graph simulations to future work.
	
	\section*{Acknowledgment}
	This work was made possible by a collaboration between Chulalongkorn University and the National University of Singapore in the project "Quantum Computing and Machine Learning for Combinatorial Optimisation" supported by the Singapore Ministry of Education Academic Research Fund (AcRF) Tier 1, grant MOE-T1-251RES2302. N.L. thanks Naphan Benchasattabuse, Patawee Prakrankamanant, and Tosaporn Angsachon for useful discussions. P.S. also thanks Pong Songpongs for useful discussions.  T.C. acknowledges the funding support from the NSRF via the Program Management Unit for Human Resources \& Institutional Development, Research and Innovation [grant number B39G680007].
	
	Lastly, we would like to express our deepest condolences for the passing of Professor Stéphane Bressan. Through his enthusiasm and vision, he tirelessly expanded our collaborative network and enriched our research. His kindness and guidance will remain an inspiration to all who worked with him, and we are grateful for the opportunity to have learned from him. May his memory live on in our continued efforts.
	
	\appendix
	\section{Manipulating Block-Encoded Operators}
	\label{appSec:BE_mani}
	
	In this appendix, we summarize how to construct block-encodings for common algebraic manipulations of operators: tensor products, products, and linear combinations. Throughout, we assume we have existing block-encodings of certain operators \(A\) and \(B\), along with appropriate ancilla registers, and we seek to form new block-encodings (e.g.\ of \(A\otimes B\), \(AB\), or \(\alpha A + \beta B\)). 
	
	\subsection{Tensor products of block-encoded matrices}
	\label{appSSec:BE_tensor}
	
	A block-encoding of \(A \otimes B\) can be assembled by juxtaposing the block-encodings of \(A\) and \(B\) and then swapping the relevant qubit registers to align with the standard block-encoding format. The following proposition is adapted from Lemma~1 in \cite{Camps_2020}.
	
	\begin{proposition}\label{prop:BE_kron}
		Let \(U_A\) be an \((\alpha,a,\epsilon_A)\)-block-encoding of an \(s\)-qubit operator \(A\) (realized at cost \(T_A\)), and let \(U_B\) be a \((\beta,b,\epsilon_B)\)-block-encoding of a \(t\)-qubit operator \(B\) (realized at cost \(T_B\)). Then
		\[
		U_{A\otimes B}
		\;=\;
		S\,\bigl(U_A \otimes U_B\bigr)\,S^\dagger
		\]
		is an \((\alpha \beta,\;a+b,\;\alpha\,\epsilon_B + \beta\,\epsilon_A + \epsilon_A\epsilon_B)\)-block-encoding of \(A \otimes B\). Here, \(S\) is defined as 
		\begin{equation}\label{eq:BE_Kron_S}
			S \;=\; \prod_{i=1}^{s} \mathrm{SWAP}_{\,a+i\,}^{\,a+b+i\,}
		\end{equation}
		where each \(\mathrm{SWAP}_{\,a+i\,}^{\,a+b+i\,}\) gate exchanges qubits \((a+i)\) and \((a+b+i)\). The circuit cost of \(U_{A\otimes B}\) is \(\mathcal{O}(T_A + T_B)\).
	\end{proposition}
	
	\begin{proof}[Proof]
		See Lemma~1 in Ref. \cite{Camps_2020} for how interleaving \(\mathrm{SWAP}\) gates and the respective block-encodings yields the desired \((\alpha\beta)\)-subnormalized block. The cost argument of \(\mathcal{O}(T_A + T_B)\) follows from Lemma~21 in Ref. \cite{Chakraborty2023}.
	\end{proof}
	
	\subsection{Products of block-encoded matrices}
	\label{appSSec:BE_prod}
	
	To multiply two operators \(A\) and \(B\) (each already block-encoded), one can ``stretch'' the ancillas so that the product \(AB\) emerges in the top-left block. Formally:
	
	\begin{proposition}\label{def:BE_prod}
		Let \(U_A\) be an \((\alpha,a,\epsilon_A)\)-block-encoding of an \(s\)-qubit operator \(A\) (cost \(T_A\)) and \(U_B\) be a \((\beta,b,\epsilon_B)\)-block-encoding of an \(s\)-qubit operator \(B\) (cost \(T_B\)). Then,
		\[
		U_{AB}
		\;=\;
		\bigl(I^{\otimes b}\!\otimes U_A\bigr)\;\bigl(I^{\otimes a}\!\otimes U_B\bigr)
		\]
		is an \(\bigl(\alpha\beta,\;a+b,\;\alpha\,\epsilon_B \;+\;\beta\,\epsilon_A\bigr)\)-block-encoding of \(AB\) and can be implemented at cost \(\mathcal{O}(T_A + T_B)\).
	\end{proposition}
	
	\begin{proof}[Proof]
		See Lemma~53 in Ref. \cite{Gily_n_2019}.
	\end{proof}
	
	\subsection{Linear combinations of unitaries (LCU)}
	\label{appSSec:BE_LCU}
	
	When we wish to form a linear combination \(\sum_{j}\gamma_j \,A_j\) of operators \(\{A_j\}\) that can be block-encoded, we can use a standard \emph{LCU construction} \cite{Gily_n_2019} that appends an ancilla register in a superposition weighted by \(\{\gamma_j\}\). The relevant coefficients \(\{\gamma_j\}\) appear via an explicit state-preparation circuit:
	
	\begin{definition}[State Preparation]\label{def:BE_statePrep}
		Let \(\vec{\gamma}\in\mathbb{C}^m\) with \(\|\vec{\gamma}\|_\infty\le \beta\). A pair of unitaries \((P_b,Q_b)\) is called a \((\beta,b,\epsilon)\)-\emph{state-preparation} for \(\vec{\gamma}\) if
		\[
		P_b \ket{0^b}
		\;=\;
		\sum_{j=0}^{2^b-1} c_j \ket{j},
		\quad
		Q_b \ket{0^b}
		\;=\;
		\sum_{j=0}^{2^b-1} d_j \ket{j},
		\]
		such that 
		\[
		\sum_{j=0}^{2^b-1}
		\bigl|
		\beta\,c_j^*\,d_j \;-\; \gamma_j
		\bigr|
		\;\le\;
		\epsilon
		\]
		and \(2^b\ge m\). Intuitively, \(\{c_j\}\) and \(\{d_j\}\) approximate the amplitudes \(\{\gamma_j/\beta\}\).
	\end{definition}
	
	By preparing such states, one can implement a block-encoding of a \emph{linear combination} of the \(\{A_j\}\) as follows.
	
	\begin{proposition}[LCU]\label{def:BE_LCU}
		Let \(B = \sum_{j=0}^{m-1} \gamma_j \,A_j\) be a linear combination of \(m\) \(s\)-qubit operators, each of which has an \((\alpha_j,a_j,\epsilon_j)\)-block-encoding \(U_{A_j}\) implemented at cost \(T_j\). Suppose we can prepare a corresponding \((\beta,b,\epsilon)\)-state-preparation of the coefficients \(\{\gamma_j\}\). Then, via the unitary
		\begin{equation}\label{eq:BE_linCom_main}
			U_B
			\;=\;
			\bigl(P_b^\dagger \!\otimes I_{\max_j a_j}\otimes I_s \bigr)\,
			W\,
			\bigl(Q_b \otimes I_{\max_j a_j}\otimes I_s\bigr),
		\end{equation}
		where
		\begin{equation}\label{eq:BE_linCom_W}
			W
			\;=\;
			\sum_{j=0}^{m-1}\!\ket{j}\!\bra{j}
			\otimes U_{A_j}
			\;+\;
			\sum_{j=m}^{2^b-1}
			\!\ket{j}\!\bra{j}
			\otimes I_{\,\max_j a_j+s},
		\end{equation}
		we obtain a \(\Bigl(\!\sum_j \gamma_j\, \alpha_j, \max_j a_j + b,\; \sum_j \gamma_j\,\epsilon_j + \epsilon\Bigr)\)-block-encoding of \(B\), realized at cost \(\mathcal{O}\!\bigl(\sum_j T_j + T_{PQ}\bigr)\), where \(T_{PQ}\) is the cost of state-preparation.
	\end{proposition}
	\begin{figure*}[ht!]
		\centering
		\includegraphics[width = 0.8\textwidth]{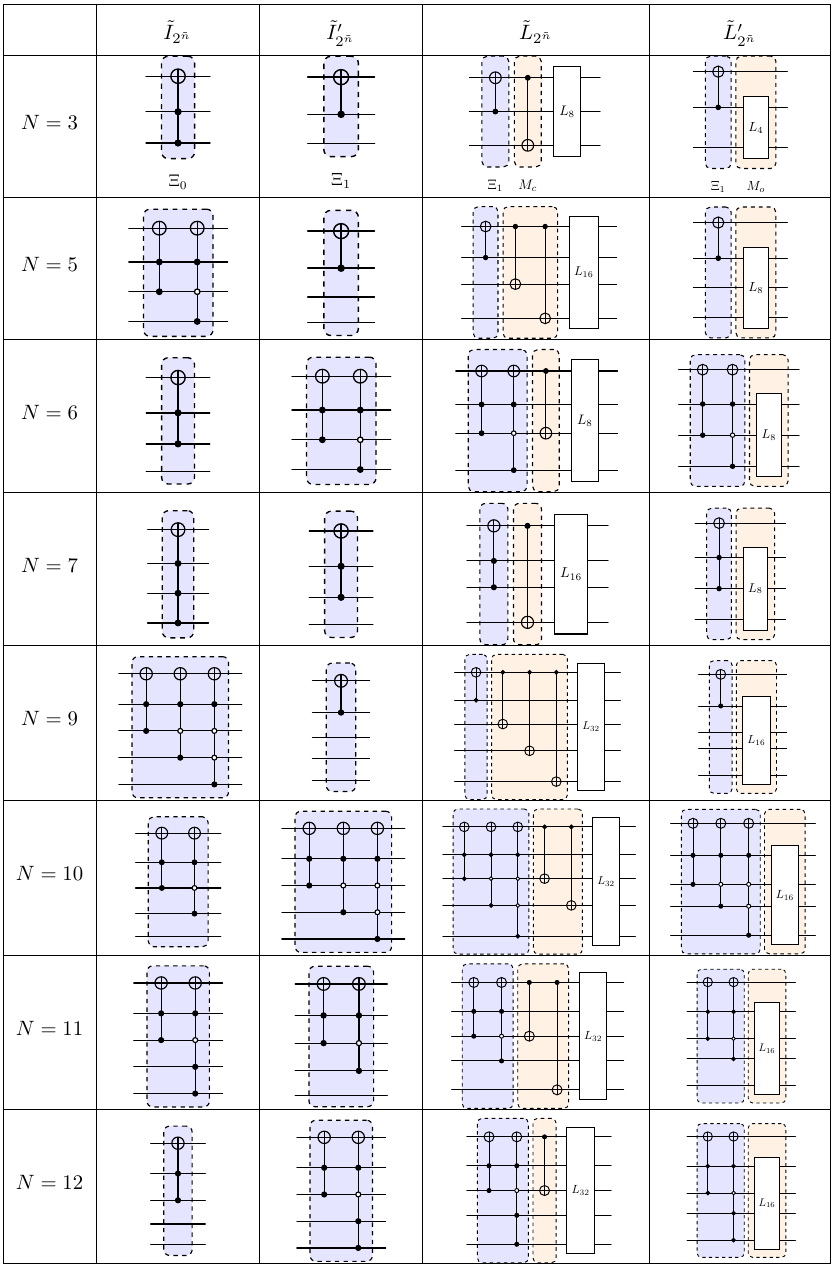}
		\caption{{\bf Examples of block-encoding circuits for \(\widetilde{I}_{2^{\tilde{n}}}\), \(\widetilde{I}'_{2^{\tilde{n}}}\), \(\widetilde{L}_{2^{\tilde{n}}}\), and \(\widetilde{L}'_{2^{\tilde{n}}}\) across a range of system sizes \(N\in\{3,5,6,7,9,10,11,12\}\).} Here \(\tilde{n} \coloneq \lceil \log_2 N\rceil\) denotes the number of index qubits.  We omit the additional LCU components (including state-preparation gates and control logic) for clarity, focusing solely on the core block-encoding subcircuits.}
		\label{fig:exam_circ_IL_BE}
	\end{figure*}
	\begin{proof}[Proof]
		See Lemma~52 in \cite{Gily_n_2019} for details. 
	\end{proof}
	
	In practice, one picks an integer \(b\) large enough that \(2^b \ge m\), then constructs unitaries \(P_b\) and \(Q_b\) satisfying the state-preparation conditions of Definition~\ref{def:BE_statePrep} (i.e., embedding the coefficient vector \(\vec{y}\) with total error at most \(\epsilon\)). The operator \(W\) in ~\eqref{eq:BE_linCom_W} can be understood as a \emph{multiplexed} selection of the block-encodings \(\{U_{A_j}\}\), each triggered by a distinct bit-string in the ancilla register. Concretely, one implements a series of (multi-)controlled gates, where each qubit in the ancilla register either serves as a ``control'' or an ``open-control'' (for logical~0), using a big-endian encoding of~\(j\). This design ensures that, conditioned on the ancillas being in state~\(\ket{j}\), the circuit applies precisely \(U_{A_j}\) on the signal and the remaining ancillas, thereby effecting the linear combination \(\sum_{j}\gamma_j\,A_j\) in the top-left block.

	\begin{figure*}[!ht]
		\centering
		\includegraphics[height=2.8in]{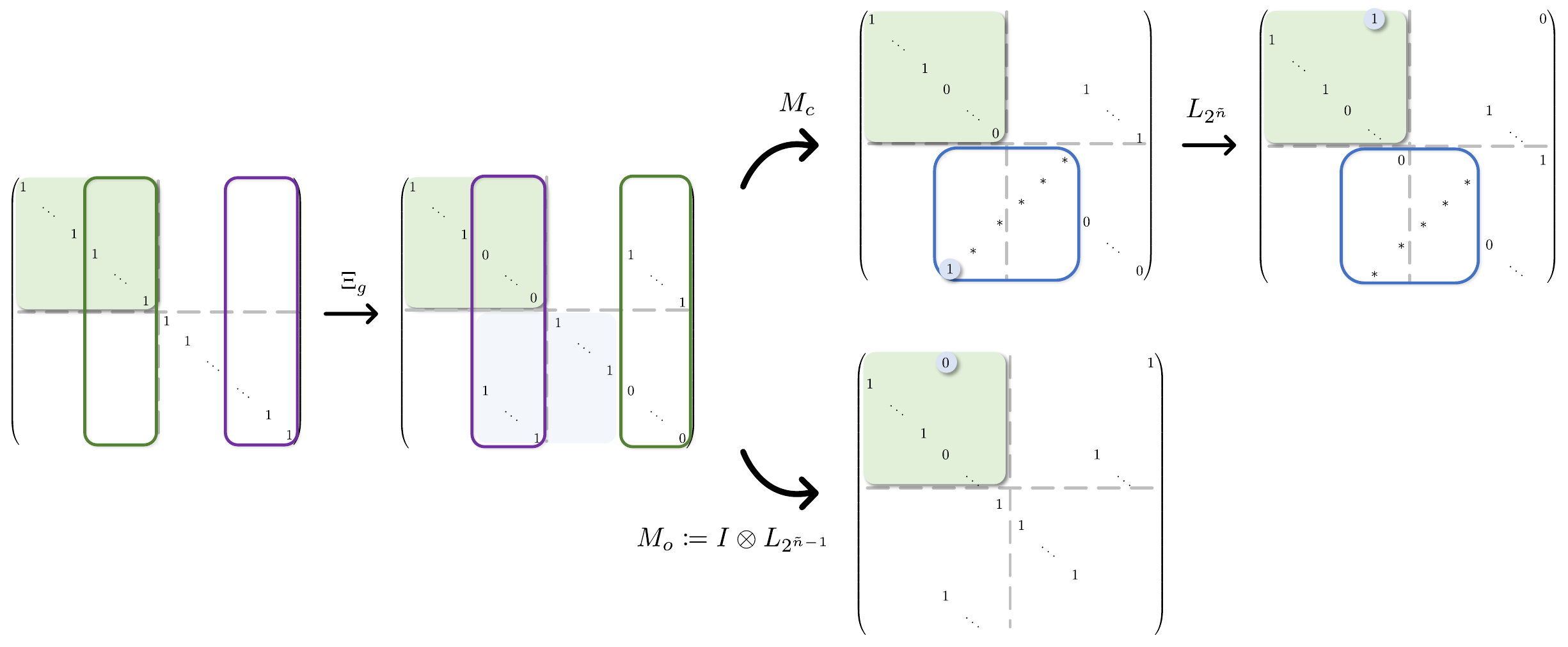}
		\caption{{\bf Matrix diagram illustrating how $\Xi_{g}$ and $M_{c,o}$ rearrange block-encoding 
				entries.} The gate $\Xi_g$ moves excess “1” entries out of the main diagonal block for both cases. For closed-chained cases, the gate $M_{c}$ repositions these “1”s near the boundary so that after $L$-shift, $\tilde{L}$ is correctly 
			realized as a block-encoding of size $2^{\tilde{n}} \!\times\! 2^{\tilde{n}}$. For open-chained cases, the gate $M_{o}$ shifts four quadrants independently so that the excess “1” will not appear in the block encoding of $\tilde{L}'$}
		\label{circ:matrixDiagram}
	\end{figure*}
	
	\begin{figure*}[!ht]
		\centering
		\subfloat[]{%
			\begin{adjustbox}{width=2.8in}
				\begin{quantikz}
					\lstick{$\ket{0}$}	&\gate{H}	&\ctrl[open]{1}		&\ctrl{1}	&\ctrl{1}	&\ctrl{1} &\gate{H}	&\\
					\lstick{$\ket{0}$}	&			&\gate[2]{\Xi_{g}}\gategroup[2, steps=1, style={dashed, rounded corners, fill = blue!10, inner xsep = 2pt}, background, label style={label position = below, anchor = north, yshift=-0.4cm}]{$U_{\tilde{I}_{2^{\tilde{n}}}}$}	&\gate[2]{\Xi_{g}}\gategroup[2, steps=3, style={dashed, rounded corners, fill = blue!10, inner xsep = 2pt}, background, label style={label position = below, anchor = north, yshift=-0.4cm}]{$U_{\tilde{L}_{2^{\tilde{n}}}}$}	&\gate[2]{M_c}&\gate[2]{L_{2^{\tilde{n}+1}}} &		&\\
					\lstick{$\ket{\psi}$}	&\qwbundle{\tilde{n}}&&&&&&
				\end{quantikz}
			\end{adjustbox}
		}
		\subfloat[]{%
			\begin{adjustbox}{width=2.5in}
				\begin{quantikz}
					\lstick{$\ket{0}$}	&\gate{H}	&\ctrl[open]{1}		&\ctrl{1}	&\ctrl{2} &\gate{H}	&\\
					\lstick{$\ket{0}$}	&			&\gate[2]{\Xi_{g}}\gategroup[2, steps=1, style={dashed, rounded corners, fill = blue!10, inner xsep = 2pt}, background, label style={label position = below, anchor = north, yshift=-0.4cm}]{\textcolor{blue}{$U_{\tilde{I'}_{2^{\tilde{n}}}}$}}	&\gate[2]{\Xi_{g}}\gategroup[2, steps=2, style={dashed, rounded corners, fill = blue!10, inner xsep = 2pt}, background, label style={label position = below, anchor = north, yshift=-0.4cm}]{\textcolor{blue}{$U_{\tilde{L'}_{2^{\tilde{n}}}}$}}	 &\gategroup[2, steps=1, style={dashed, rounded corners, fill = red!10, inner xsep = 0.01pt}, background, label style={label position = above, anchor = north, yshift=-0.15cm, xshift=0.28cm}]{\textcolor{red!80}{$M_o$}}		& &\\
					\lstick{$\ket{\psi}$}	&\qwbundle{\tilde{n}}&&&\gate[1]{L_{2^{\tilde{n}}}}&&
				\end{quantikz}
			\end{adjustbox}
		}
		\caption{ {\bf Panels (a) and (b) illustrate $(2,2,0)$--block-encoding circuits 
				for $\tilde{\Phi}_F^{c}$ (closed-chained) and $\tilde{\Phi}_F^{o}$ (open-chained), 
				respectively.} In both cases, the operator $\Xi_{g}$ removes unwanted ones from the top-left block, which is used for every partial blocks of $\tilde{I}_{2^{\tilde{n}}} (\tilde{I}_{2^{\tilde{n}}}')$ and $\tilde{L}_{2^{\tilde{n}}} (\tilde{L}_{2^{\tilde{n}}}')$. For the closed-chained systems, the operator $M_{c}$ adjusts boundary values according to $\mathfrak{C}=1$ so that applying $L_{\tilde{n}+1}$ completes the shift 
			to ensure $\tilde{\Phi}_F$ matches the desired incidence structure. For open-chained systems, 
			$M_{o} \coloneq I \otimes L_{2^{\tilde{n}}}$ is used, with no additional $L$-shift operation required.}
		\label{circ:gen2_phiF}
	\end{figure*}
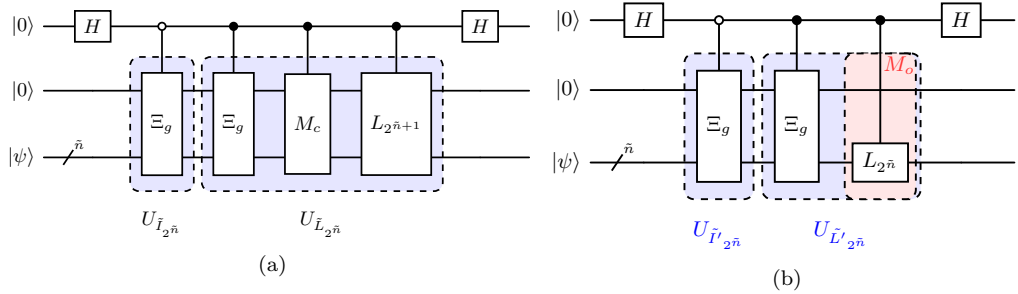
	
	\section{Details of Block Encoding Construction for Generalized Coupled Oscillators}
	\label{app:genMKnoOsc_CCO}
	
	We now relax the condition that the total number of oscillators must be a power of two, 
	allowing $N$ to be any positive integer with $N \neq 2^n$. 
	To accommodate a general $N$, we embed the spring–mass system’s matrix $\mathcal{A}$ (and likewise $F$) 
	into a larger $(2^{\tilde{n}} \!\times\! 2^{\tilde{n}})$ matrix with additional zero rows and columns, 
	where $\tilde{n} \coloneqq \lceil \log_2 N \rceil$. Concretely, we define 
	\begin{equation}\label{eq:A_equiMK_withZero}
		\tilde{\Phi}_F^{c/o} \;=\;
		\begin{pmatrix}
			\mathcal{A}_{N\times N}(\mathfrak{C}) & 0\\
			0&0\\
		\end{pmatrix}_{2^{\tilde{n}} \times 2^{\tilde{n}}},
	\end{equation}
	where $\mathfrak{C} = 1$ (closed-chained, i.e.\ periodic boundary condition) or $0$ (open-chained, i.e. \ open boundary condition), 
	and the extra rows/columns are purely zero. We label such zero-padded matrices by a \emph{tilde}, 
	e.g.\ $\tilde{F}$ or $\tilde{\Phi}_F$, to distinguish them from the original size-$N$ system.
	
	In this padded setting, the incidence matrix $\tilde{\Phi}_F$ (analogous to $\Phi_F$ in 
	Sec.~\ref{sss:genMK_CCO} but extended to $2^{\tilde{n}}$ size) decomposes similarly:
	\begin{equation}\label{NEW_eq:incFcy_decomOG_withZero}
		\tilde{\Phi}_F^{c} 
		\;=\;
		\tilde{I}_{2^{\tilde n}} - \tilde{L}_{2^{\tilde n} }
		\;=\;
		\begin{pmatrix}
			I_N & 0\\
			0 & 0
		\end{pmatrix}_{2^{\tilde{n}} \times 2^{\tilde{n}}}
		\;-\;
		\begin{pmatrix}
			L_N & 0\\
			0 & 0
		\end{pmatrix}_{2^{\tilde{n}} \times 2^{\tilde{n}}},
	\end{equation}
	for a closed-chained system (where $I_N$ and $L_N$ are from~\eqref{eq:decom_cy_B}), 
	and
	\begin{equation}\label{NEW_eq:incFncy_decomOG_withZero}
		\tilde{\Phi}_F^{o} 
		\;=\;
		\tilde{I}'_{2^{\tilde n}} - \tilde{L}'_{2^{\tilde n}}
		\;=\;
		\begin{pmatrix}
			I'_N & 0\\
			0 & 0
		\end{pmatrix}_{2^{\tilde{n}} \times 2^{\tilde{n}}}
		\;-\;
		\begin{pmatrix}
			L'_N & 0\\
			0 & 0
		\end{pmatrix}_{2^{\tilde{n}} \times 2^{\tilde{n}}},
	\end{equation}
	for an open-chained system, where the {\it prime} notation again signifies that 
	the last column is zeroed out as in~\eqref{NEW_eq:incFcy_decomOG_withZero}.
	
	\paragraph*{Circuit construction for block-encoding of \boldmath\(B\).}
	Figures~\ref{circ:BE_diag}--\ref{circ:BE_genB_ncy} (and Algorithm~\ref{alg:cap}) 
	extend directly to these padded matrices: one composes the extended matrices  
	\(\tilde{I}_{2^{\tilde{n}}}\)-, \(\tilde{I}'_{2^{\tilde{n}}}\)-, \(\tilde{L}_{2^{\tilde{n}}}\)-, or \(\tilde{L}'_{2^{\tilde{n}}}\)-block-encodings, 
	using our proposed multi-controlled gates $\Xi_g$ and $M_{o(c)}$ depending on the given boundary conditions. 
	The final step aligns with the same linear-combination-of-unitaries (LCU) approach outlined 
	in the $N=2^n$ scenario.

	Algorithm~\ref{alg:cap} thoroughly provides the step-by-step procedure to implement the four aforementioned extended matrices, leaving out unneeded elements outside the main block. A short summary of this algorithm is provided below.
	
	\paragraph*{Brief Summary of Algorithm \ref{alg:cap}.}
	To implement a block encoding of any operator in
	\(\{\tilde{I}_{2^{\tilde{n}}},\,\tilde{I}'_{2^{\tilde{n}}},\,\tilde{L}_{2^{\tilde{n}}},\,\tilde{L}'_{2^{\tilde{n}}}\}\),
	we follow the steps illustrated in Algorithm~\ref{alg:cap}:
	
	\begin{enumerate}[label=(\roman*)]
		\item \textbf{Block encoding of \(\tilde{I}_{2^{\tilde{n}}}\) or \(\tilde{I}'_{2^{\tilde{n}}}\).}  
		We utilize the proposed \(\Xi_{g}\) gate to shift unwanted ``1'' entries outside the desired block. 
		The choice of \(\Xi_{g}\) depends on the target operator: 
		\(\Xi_{0}\) is dedicated to \(\tilde{I}_{2^{\tilde{n}}}\) alone, 
		while \(\Xi_{1}\) is used for \(\tilde{I}'_{2^{\tilde{n}}},\;\tilde{L}_{2^{\tilde{n}}},\;\tilde{L}'_{2^{\tilde{n}}}\).
		
		\item \textbf{Block encoding of \(\tilde{L}_{2^{\tilde{n}}}\) or \(\tilde{L}'_{2^{\tilde{n}}}\).}  
		First, apply \(\Xi_{g}\) (as above) to position the active entries appropriately. 
		In the \emph{closed-chained} case, we then use the gate \(M_{c}\) 
		to move the ``1'' to the correct bottommost position, 
		followed by the \(L\)-shift matrix to obtain the final \(\tilde{L}_{2^{\tilde{n}}}\).  
		In the \emph{open-chained} case, we directly use \(M_{o} = I \otimes L_{2^{\tilde{n}-1}}\) 
		to obtain the block encoding of \(\tilde{L}'_{2^{\tilde{n}}}\).
	\end{enumerate}
	
	As in the $N=2^n$ case, we incorporate these partial gates 
	(\(\tilde{I}_{2^{\tilde{n}}},\tilde{L}_{2^{\tilde{n}}}\) for closed-chained or \(\tilde{I}'_{2^{\tilde{n}}},\tilde{L}'_{2^{\tilde{n}}}\) for open-chained) 
	into a single linear-combination circuit, acquiring $\tilde{\Phi}_F^{c/o}$, as illustrated in Fig. \ref{circ:gen2_phiF}.
	
	Finally, by combining $\tilde{\Phi}_F^{c/o}$ with $\sqrt{M}^{-1}$ and $\sqrt{W}$, 
	one obtains the required block-encoding of $B$ for an $N$-oscillator system, 
	\emph{even when $N \neq 2^n$}. All subsequent Hamiltonian-simulation steps, from QSVT to amplitude amplification, 
	remain unchanged. Some examples of the block encodings of the sub-matrices $\tilde{I}_{2^{\tilde{n}}}$ , $\tilde{I}'_{2^{\tilde{n}}}$ , $\tilde{L}_{2^{\tilde{n}}}$ and $\tilde{L}'_{2^{\tilde{n}}}$ are given in Fig. \ref{fig:exam_circ_IL_BE}. For simplicity, we ignore the LCU components, that is, Hadamard gates and control nodes.
	
	\begin{algorithm*}[ht!]
		\DontPrintSemicolon
		\caption{Circuit construction for block-encoding \(\tilde{I}\) and \(\tilde{L}\) in both open- and closed-chained cases.}
		\label{alg:cap}
		
		\KwIn{Total number of oscillators \(N\) and \texttt{case} $\in\{\tilde{I}_{2^{\tilde n}},\,\tilde{L}_{2^{\tilde n}},\,\tilde{I}'_{2^{\tilde n}},\,\tilde{L}'_{2^{\tilde n}}\}$.}
		\KwOut{The final circuit \(\mathcal{U}\), which serves as a block-encoding of the chosen 
			\(\tilde{I}_{2^{\tilde n}}\), \(\tilde{L}_{2^{\tilde n}}\), \(\tilde{I}'_{2^{\tilde n}}\), or \(\tilde{L}'_{2^{\tilde n}}\).}
		\vspace{5pt}
		
		\textbf{Initialization:}
		\begin{enumerate}
			\item Let ${\tilde{n}} \coloneqq \bigl\lceil\log_2 N\bigr\rceil$.
			\item Prepare an empty circuit \(\mathcal{U}\) with \(({\tilde{n}}+1)\) qubits where the top qubit (qubit $0$) being an ancilla and the remaining \({\tilde{n}}\) encoding indices.
			\item Define a parameter \(g\) based on the chosen case:
			\begin{equation*}
				g \;=\;
				\begin{cases}
					0, & \text{if \texttt{case}} = \tilde{I}_{2^{\tilde n}},\\
					1, & \text{if \texttt{case}} \in \{\tilde{L}_{2^{\tilde n}}, \tilde{I}'_{2^{\tilde n}}, \tilde{L}'_{2^{\tilde n}}\}.
				\end{cases}
			\end{equation*}
		\end{enumerate}
		\smallskip
		
		\textbf{Step 1: Implement the \(\Xi_{g}\) gate.}
		\begin{itemize}
			\item Define 
			\[
			\Xi_{g}
			\;\coloneqq\;
			\sum_{i=0}^{N-g-1}\!\dyad{i}{i}
			\;+\;\!
			\sum_{i=N-g}^{2^{\tilde{n}}-1}\!\dyad{i+2^{\tilde{n}}}{i}
			\;+\;\!
			\sum_{i=2^{\tilde{n}}}^{\,2^{\tilde{n}}+N-g-1}\!\dyad{i}{i}
			\;+\;\!
			\sum_{i=2^{\tilde{n}}+N-g}^{2^{{\tilde{n}}+1}-1}\!\dyad{i-2^{\tilde{n}}}{i}.
			\]
			(Conceptually, \(\Xi_g\) \textit{shifts} certain basis states between the 
			upper and lower halves of the ancilla-index space so that the top-left block 
			becomes an identity (or near-identity).)
			\item In practice, implement \(\Xi_g\) with repeated multi-controlled $\mathrm{NOT}$ gates:
			\begin{enumerate}
				\item Let \(h = 2^{\tilde{n}}\). 
				\item Maintain an index list \(\texttt{notCtrlQubits}=\{\}\).
				\item \textbf{\texttt{for}} \(i=1\) to \({\tilde{n}}\):
				\begin{itemize}
					\item[] \textbf{\texttt{if}} \(\;h - 2^{\,{\tilde{n}}-i} \,\ge\, N - g\), do:
					\begin{enumerate}
						\item Apply \(\mathrm{NOT}\) on all qubits in \(\texttt{notCtrlQubits}\) (if non-empty).
						\item Apply an multi-controlled NOT gate with control qubits on \(\{1,\ldots,i\}\) 
						and target on qubit \(0\).
						\item Re-apply \(\mathrm{NOT}\) on those same qubits in \(\texttt{notCtrlQubits}\).
						\item Append qubit \(i\) to \(\texttt{notCtrlQubits}\).
						\item Update \(h\leftarrow h - 2^{\,{\tilde{n}}-i}\).
					\end{enumerate}
				\end{itemize}
			\end{enumerate}
		\end{itemize}
		
		\textbf{Step 2: Add \(\,M_c\) or \(\,M_o\) gates.}
		\begin{itemize}
			\item \textbf{\texttt{if}} \(\texttt{case}=\tilde{L}_{2^{\tilde n}}\), define
			\(
			M_{c}
			\;=\;
			\sum_{i=0}^{\,2^{\tilde{n}}+N-g-1}\!
			\dyad{i}{i}
			\;+\;
			\dyad{\,2^{n+1}-1}{\,2^{\tilde{n}}+N-g}
			\;+\;
			\dyad{\perp}{\perp}
			\), do:
			\begin{enumerate}
				\item Convert $2^{\tilde{n}} + (N-g)$ into its binary form and toggle specified qubits using CNOT gates with control on qubit 0 and target on $i$th qubit if the value in the bitstring at index $i$ is 0
				\item Apply an $L_{{\tilde{n}}+1}$-shift gate.
			\end{enumerate}
			
			\item \textbf{\texttt{else if}} \(\texttt{case}=\tilde{L}'_{2^{\tilde n}}\), apply $M_o = I \otimes L_{2^{\tilde{n}-1}}$ where $L_{{\tilde{n}}}$-shift gate is operating on remaining $\tilde{n}$ qubits.
			
		\end{itemize}
	\end{algorithm*}
	

	\newpage
	\bibliography{apssamp}
	
\end{document}